\documentclass[a4paper]{amsart}

\usepackage[T1]{fontenc}
\usepackage{lmodern}

\usepackage{amssymb, amsfonts, amsthm, mathtools, tikz-cd, graphicx}
\tikzcdset{
  row sep/normal=1.35em,
  column sep/normal=1.8em
}
\usepackage{thmtools, thm-restate}
\usepackage[shortlabels]{enumitem}
\definecolor{colorblindblue}{RGB}{17,112,170}
\definecolor{colorblindorange}{RGB}{200,82,0}
\usepackage[linktocpage]{hyperref}
\hypersetup{
  colorlinks   = true,
  urlcolor     = colorblindblue,
  linkcolor    = colorblindblue,
  citecolor   = colorblindblue}
\usepackage[capitalise,noabbrev]{cleveref}
\usepackage[font=small,skip=0pt]{caption}

\makeatletter
\AddToHook{cmd/appendix/before}{\def\cref@section@alias{appendix}\def\cref@subsection@alias{appendix}}
\makeatother

\usepackage[nocompress,noadjust]{cite}

\setitemize[1]{label=\raisebox{0.25ex}{\tiny$\bullet$}}

\newtheorem{theorem}{Theorem}[section]
\newtheorem{lemma}[theorem]{Lemma}
\newtheorem{proposition}[theorem]{Proposition}
\theoremstyle{definition}
\newtheorem{definition}[theorem]{Definition}

\renewcommand{\epsilon}{\varepsilon}

\providecommand\given{}
\newcommand\SetSymbol[1][]{%
  \nonscript\:#1\vert{}
  \allowbreak{}
  \nonscript\:
  \mathopen{}}
\DeclarePairedDelimiterX\Set[1]\{\}{%
  \renewcommand\given{\SetSymbol[]}
  #1
}

\newcommand{\nospaceperiod}{\makebox[0pt][l]{\,.}}

\newcommand{\R}{\mathbb{R}}
\newcommand{\N}{\mathbb{N}}
\newcommand{\norm}[1]{\left\lVert#1\right\rVert}
\DeclarePairedDelimiter\abs{\lvert}{\rvert}
\newcommand{\isomorphic}{\ensuremath{\cong}}

\newcommand{\opt}{\mathrm{opt}}

\newcommand{\deff}{\emph}

\DeclareMathOperator{\Cov}{Cov}
\DeclareMathOperator{\SCov}{SCov}
\DeclareMathOperator{\ins}{ins}
\DeclareMathOperator{\slow}{slow}
\DeclareMathOperator{\SB}{SB}
\newcommand{\Cont}{\mathcal{C}}
\DeclareMathOperator{\Cech}{\v{C}}
\DeclareMathOperator{\SCech}{S\v{C}}
\DeclareMathOperator{\SSub}{SSub}
\DeclareMathOperator{\Sub}{Sub}
\DeclareMathOperator{\Rips}{Rips}

\newcommand{\reach}{\Gamma}

\newcommand{\order}{\Delta}

\makeatletter
\newcommand*\rel@kern[1]{\kern#1\dimexpr\macc@kerna}
\newcommand*\widebar[1]{%
  \begingroup
  \def\mathaccent##1##2{%
    \rel@kern{0.8}%
    \overline{\rel@kern{-0.8}\macc@nucleus\rel@kern{0.2}}%
    \rel@kern{-0.2}%
  }%
  \macc@depth\@ne
  \let\math@bgroup\@empty \let\math@egroup\macc@set@skewchar
  \mathsurround\z@ \frozen@everymath{\mathgroup\macc@group\relax}%
  \macc@set@skewchar\relax
  \let\mathaccentV\macc@nested@a
  \macc@nested@a\relax111{#1}%
  \endgroup
}
\makeatother

\makeatletter
\@namedef{r@tocindent4}{0pt}
\@namedef{r@tocindent5}{0pt}
\makeatother
\makeatletter
\renewcommand\subparagraph{\@startsection{subparagraph}{5}%
  \z@{.5\linespacing\@plus.7\linespacing}{-.5em}%
  {\normalfont\bfseries}}
\makeatother

\newcommand{\downset}[2]{#1 \downarrow #2}

\title{A Sparse Multicover Bifiltration of Linear Size}
\author{Ángel Javier Alonso}
\address{Institute of Geometry, Graz University of Technology, Austria}
\email{alonsohernandez@tugraz.at}

\begin{document}

\begin{abstract}
  The $k$-cover of a point cloud $X$ in $\R^{d}$ at radius $r$ is the set of all
  points within distance $r$ of at least $k$ points of $X$. By varying $r$
  and $k$ we obtain a two-parameter filtration known as the multicover
  bifiltration. This bifiltration has received attention recently due to being
  choice-free and robust to outliers. However, it is hard to compute: the
  smallest known equivalent simplicial bifiltration has $O(\abs{X}^{d+1})$
  simplices. We introduce a $(1+\epsilon)$-approximation of the
  multicover bifiltration of linear size $O(\abs{X})$, for fixed $d$ and
  $\epsilon$. The methods also apply to the subdivision Rips bifiltration on
  metric spaces of bounded doubling dimension, yielding analogous results.
\end{abstract}

\maketitle

\section{Introduction}

This paper aims to approximate the multicover bifiltration of a finite subset
$X$ of $\R^{d}$. The \deff{$k$-cover} $\Cov(r,k)$ for a scale $r\geq 0$ is
given by all points covered by at least $k$ balls of radius $r$ around the
points of $X$:
\begin{equation*}
  \Cov(r,k)\coloneqq \Set{p\in\R^{d} \given \norm{x - p} \leq r \text{ for at least $k$ points $x\in X$}}.
\end{equation*}
For any $r \leq r'$ and $k \geq k'$, we have that
$\Cov(r,k)\subset \Cov(r',k')$ and as such the $\Cov(r,k)$ assemble into a
\deff{bifiltration} known as the \deff{multicover bifiltration}, see~\cref{fig:multicover}.
\begin{figure}
  \centering
  \includegraphics[interpolate]{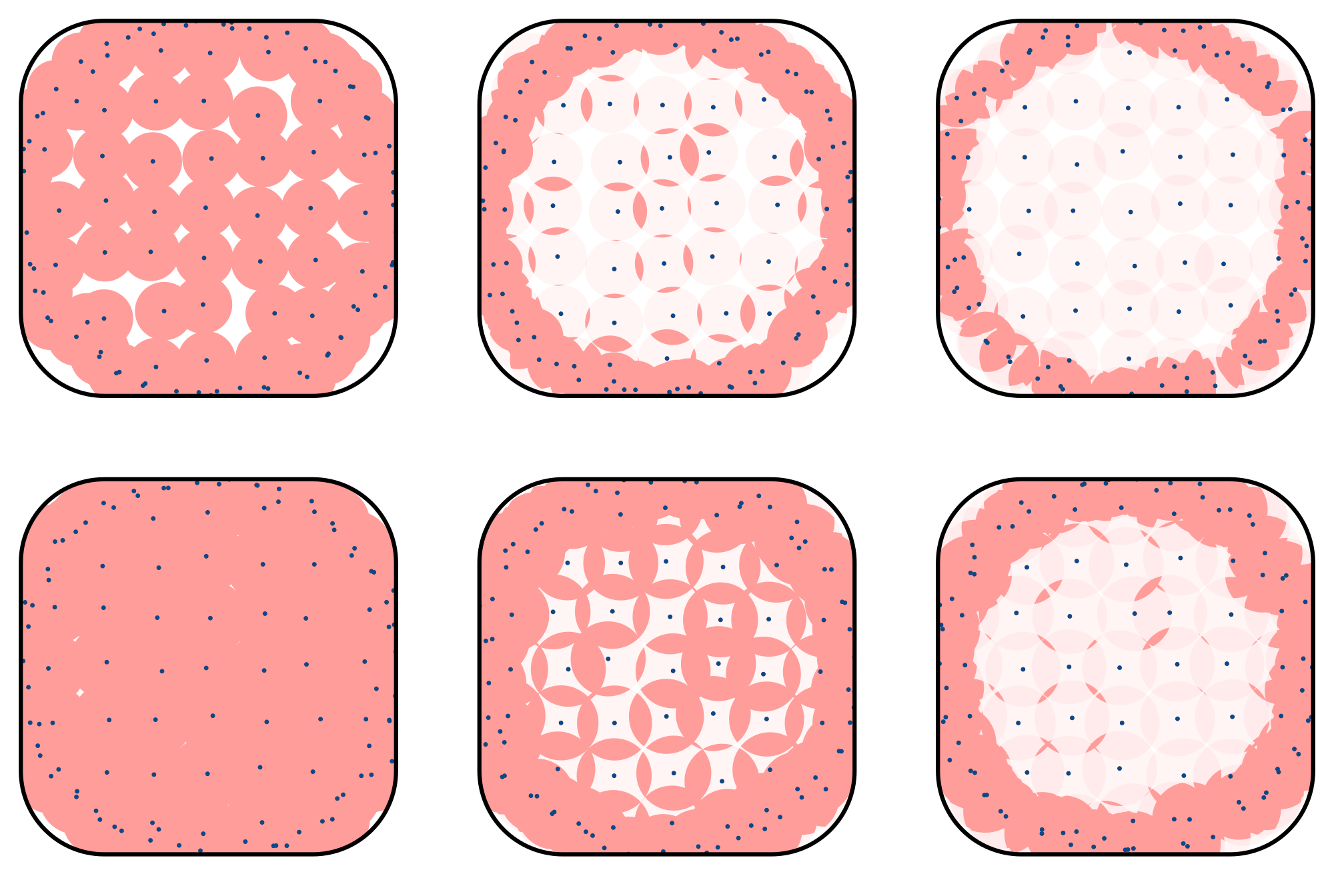}
  \caption{Illustration of the multicover bifiltration $\Cov(r,k)$ for various
    values of $r$ and $k$. The left column has $k = 1$ and is thus the union of
    balls around the points, the middle column shows those points covered by at
    least two balls, $k=2$, and the right column has $k = 3$: those points
    covered by at least three balls.
    The top row has a
    smaller scale parameter $r$ than the bottom row.
  }\label{fig:multicover}
\end{figure}

The multicover bifiltration offers a view into the topology of the point cloud
$X$ across multiple scales, in a way that is robust to
outliers~\cite{blumbergStability2ParameterPersistent2024}. However, computing
the multicover bifiltration remains a challenge. Current exact combinatorial
(simplicial or polyhedral)
methods~\cite{edelsbrunnerMultiCoverPersistenceEuclidean2021,corbetComputingMulticoverBifiltration2023}
are expensive, of size $O(\abs{X}^{d+1})$, and previous approximate
methods~\cite{buchetSparseHigherOrder2024,buchetSparseHigherOrder2023} are in general exponential and
only linear in $\abs{X}$ when taking $k$ to be at most a constant.

For general metric spaces, instead of subsets of $\R^{d}$, an analogue of the
multicover is the subdivision Rips
bifiltration~\cite{sheehyMulticoverNerveGeometric2012}. Like the multicover, it has received
attention recently, both in its
theoretical~\cite{blumbergStability2ParameterPersistent2024} and computational
aspects~\cite{lesnickNerveModelsSubdivision2024,lesnickSparseApproximationSubdivisionRips2024,hellmerDensitySensitiveBifiltered2024},
where it faces similar challenges. Previous approximation
schemes~\cite{lesnickInteractiveVisualization2D2015,blumbergStability2ParameterPersistent2024,lesnickNerveModelsSubdivision2024,lesnickSparseApproximationSubdivisionRips2024,hellmerDensitySensitiveBifiltered2024}
are of size polynomial (but not linear) in $\abs{X}$; see the related work section below.

In this paper, we define a simplicial bifiltration, the \deff{sparse subdivision
  bifiltration}~(\cref{def:sparse_subdivision}), that
$(1 + \epsilon)$-approximates the multicover and that is of linear size
$O(\abs{X})$, independently of $k$ and for fixed $\epsilon$ and $d$. It can be
computed in time $O(\abs{X} \log\Delta)$, where $\Delta$ is the \deff{spread} of
$X$, the ratio between the longest and shortest distances between points in $X$.
In addition, the methods extend to obtain analogous results for the subdivision
Rips bifiltration for metric spaces of bounded doubling
dimension, as we discuss in~\cref{sec:rips}.

The sparse subdivision bifiltration is homotopically equivalent to the \deff{sparse
  multicover bifiltration} $\SCov$ (\cref{def:sparse_multicover}), a
bifiltration of subsets of $\R^{d}$. It is via $\SCov$ that the sparse
subdivision bifiltration $(1+\epsilon)$-approximates the multicover, by which we
mean that
\begin{equation*}
  \Cov(r,k) \subset \SCov((1+3\epsilon)r,k) \text{ and } \SCov(r,k)\subset \Cov((1+\epsilon)r,k)
\end{equation*}
for all $r$ and $k$, as shown in~\cref{thm:approximation}. In other words, the sparse subdivision
bifiltration and the multicover are \deff{(multiplicatively)
  $(1+\epsilon)$-homotopy-interleaved}~\cite{blumbergUniversalityHomotopyInterleaving2023}.

This approximation is analogous to the one obtained by Cavanna, Jahanseir, and
Sheehy~\cite{cavannaGeometricPerspectiveSparse2015} for the case $k=1$, and our
techniques can be traced back to their work; see the related work section.
A key principle is that instead of balls of radius $r$, the
sparse multicover bifiltration is based on \deff{sparse balls}
(\cref{def:sparse_ball}), understood as balls with a lifetime of three phases:
in the first phase they grow normally (at scale $r$ they have radius $r$), and
then they are slowed down (they keep growing but at scale $r$ they have radius
less than $r$), before eventually disappearing. Cavanna, Jahanseir and
Sheehy~\cite{cavannaGeometricPerspectiveSparse2015} (there is a video~\cite{cavannaVisualizingSparseFiltrations2015}) do the same, with technical differences that include stopping balls from
growing, rather than growing slowly. The size bounds here are also an adaptation
of their methods, as well as the technique of lifting the construction one
dimension higher, as in the \emph{cones} of~\cref{sec:topological_equivalence}.

We guarantee that when a sparse ball disappears it is approximately covered by
another present sparse ball, and we keep track of which present sparse ball
covers which disappeared sparse ball via a \deff{covering
  map}~(\cref{def:covering_map}). A sparse ball is weighted by as many sparse
balls it approximately covers according to the covering map---the \deff{sparse
  $k$-cover} $\SCov(r,k)$ consists of those points covered by sparse balls
whose total weight is at least $k$.

\subsection{Motivation and related work}
For $k = 1$, the $1$-cover is the union of balls of radius $r$ around the points of
$X$. The union of balls is commonly used in reconstructing submanifolds from
samples~\cite{niyogiFindingHomologySubmanifolds2008,attaliVietorisRipsComplexes2013,attaliVietorisripsComplexesAlso2011}
and is a cornerstone of \textit{persistent homology} methods in topological data
analysis, being homotopy equivalent to the \v{C}ech complex and the alpha
complex~\cite{edelsbrunnerUnionBallsIts1995,edelsbrunnerComputationalTopologyIntroduction2010},
and related to the (Vietoris-)Rips complex. The union of balls assembles into a $1$-parameter filtration:
increasing the scale parameter $r$ yields inclusions
$\Cov(r, 1)\hookrightarrow \Cov(r', 1)$, for $r\leq r'$. Taking the homology
of each $\Cov(r,1)$, we obtain a \deff{persistence module}, an algebraic
descriptor of the topology of the filtration across the different scales. Such a
persistence module is stable to perturbations of the
points~\cite{cohen-steinerStabilityPersistenceDiagrams2007,cohen-steinerStabilityPersistenceDiagrams2005}, but is not robust
to outliers (already appreciated in the first column of~\cref{fig:multicover}) and insensitive to
differences in density in the point cloud.

There are multiple methods that address the lack of robustness to outliers and
changes in density within the $1$-parameter framework. These include density
estimation~\cite{phillipsGeometricInferenceKernel2015,chazalPersistenceBasedClusteringRiemannian2013,bobrowskiTopologicalConsistencyKernel2017},
distance to a
measure~\cite{chazalGeometricInferenceProbability2011,guibasWitnessedKDistance2013,guibasWitnessedKdistance2011,chazalRobustTopologicalInference2017,anaiDTMbasedFiltrations2020,anaiDTMBasedFiltrations2019,buchetEfficientRobustPersistent2016,buchetEfficientRobustPersistent2015},
and subsampling~\cite{blumbergRobustStatisticsHypothesis2014}. However, they all
depend on choosing a parameter. The problem is that it is not clear how to
choose such a parameter for all cases, and such a choice might focus on a
specific range of scales or densities. We refer to~\cite[Section
1.7]{blumbergStability2ParameterPersistent2024} for a complete overview of the
methods and their limitations.

It is then natural to consider constructions over two parameters, scale and
density. Examples include the density
bifiltration~\cite{carlssonTheoryMultidimensionalPersistence2009}, the degree
bifiltration~\cite{lesnickInteractiveVisualization2D2015}, and the
multicover, which is closely related to both the distance to a
measure~\cite{chazalGeometricInferenceProbability2011} and $k$-nearest
neighbors~\cite{sheehyMulticoverNerveGeometric2012}. The advantage of the
multicover is that it does not depend on any further choices (like choosing a
density estimation function) and that it is robust to outliers, as a consequence
of its strong stability
properties~\cite{blumbergStability2ParameterPersistent2024}.

However, one currently problematic aspect of the multicover is its computation.
Sheehy~\cite{sheehyMulticoverNerveGeometric2012} introduced an influential
simplicial model of the multicover called the \textit{subdivision (\v{C}ech)
  bifiltration}, based on the barycentric subdivision of the \v{C}ech
filtration. It has exponentially many vertices in the number of input points,
making its computation infeasible. A crucial ingredient in the theory is the
\textit{multicover nerve
  theorem}~\cite{sheehyMulticoverNerveGeometric2012,cavannaWhenWhyTopological2017,blumbergStability2ParameterPersistent2024},
which establishes the topological equivalence
(see~\cref{sec:topological_equivalence} for a precise definition) of the
subdivision \v{C}ech and multicover bifiltrations. Such a multicover nerve
theorem has its analogue for the sparse multicover: the \textit{sparse
  multicover nerve theorem}
(\cref{thm:sparse_multicover_nerve}).

Dually to a hyperplane arrangement in $\R^{d+1}$, Edelsbrunner and
Osang~\cite{edelsbrunnerMultiCoverPersistenceEuclidean2021,edelsbrunnerMulticoverPersistenceEuclidean2018conf}
define the~\textit{rhomboid tiling} and use it to compute, up to homotopy, slices (fixing $k$ and
varying $r$ or vice versa) of the multicover
bifiltration. Corbet, Kerber, Lesnick and
Osang~\cite{corbetComputingMulticoverBifiltration2023,corbetComputingMulticoverBifiltration2021conf}
build on it to define two bifiltrations topologically
equivalent to the multicover, of size $O(\abs{X}^{d+1})$.

Buchet, Dornelas, and Kerber~\cite{buchetSparseHigherOrder2024}
introduced an elegant $(1+\epsilon)$-approximation of the multicover whose
$m$-skeleton has size that is linear in $\abs{X}$ but that incurs an exponential
dependency on the maximum order $k$ we want to compute. A crucial difference
between our work and the Buchet-Dornelas-Kerber (BDK) sparsification is that BDK
work at the level of intersection of balls, while we work directly at the level of balls.
This starting point is what allows us to ultimately obtain a construction of
linear size, independently of $k$.

In addition, BDK work by freezing the lenses: at a certain scale they stop growing. It is not
clear (or is technically challenging, in the words of BDK) how to compute exactly the first scale at which freezing lenses first
intersect---a problem they sidestep by discretizing the scale parameter (at no
complexity cost). In contrast, by letting sparse balls grow slowly at a certain
rate, rather than freezing them completely, we can compute their first
intersection time exactly, as we explain in~\cref{sec:intersection_balls}.

\subparagraph{Sparsification} Our
methods are part of a line that can be traced to the seminal work of
Sheehy~\cite{sheehyLinearSizeApproximationsVietoris2013,sheehyLinearsizeApproximationsVietorisrips2012} to obtain a linear size
approximation of the Vietoris-Rips filtration. Subsequently, Cavanna, Jahanseir
and Sheehy~\cite{cavannaGeometricPerspectiveSparse2015} simplified and
generalized Sheehy's methods to obtain a linear size $(1+\epsilon)$-approximation of the
union-of-balls filtration $\Cov(r,1)$. They construct a filtration $S$ such that
  $S(r)\subset \Cov(r, 1)\subset S((1+\epsilon)r)$, resembling~\cref{thm:approximation} here. Moreover, the \textit{sparse balls} we use here are
  directly inspired by their methods: they use balls that grow normally, stop
  growing, and eventually disappear.

These methods, broadly referred to as \textit{sparsification}, have also been applied to the Delaunay
triangulation~\cite{sheehySparseDelaunayFiltration2021} and, as already
mentioned, the multicover itself~\cite{buchetSparseHigherOrder2024}. A
fundamental ingredient in many of these constructions is the greedy permutation~\cite{rosenkrantzAnalysisSeveralHeuristics1977,gonzalezClusteringMinimizeMaximum1985,dyerSimpleHeuristicPcentre1985},
or variants of it, which we also use in the form of~\deff{persistent nets}
(\cref{sec:persistence_nets}).

\subparagraph{General metric spaces} In general metric spaces, an analogue of
the subdivision \v{C}ech bifiltration of Sheehy is the \deff{subdivision Rips
  bifiltration}. Its computation faces similar challenges, and, in fact, no
subexponential size simplicial model of it
exists~\cite{lesnickNerveModelsSubdivision2024}. Thus, there has been recent
interest in approximations of subdivision Rips. Indeed, the \deff{degree Rips
  bifiltration}~\cite{lesnickInteractiveVisualization2D2015} can be shown to be
a $\sqrt{3}$-approximation~\cite{blumbergStability2ParameterPersistent2024},
whose $m$-skeleton has size $O(\abs{X}^{m+2})$, and has been
implemented~\cite{lesnickInteractiveVisualization2D2015,rolleStableConsistentDensityBased2024,scoccolaPersistablePersistentStable2023}.
Furthermore, it was recently shown that subdivision Rips admits
$\sqrt{2}$-approximations whose $k$-skeleta have the same size as degree
Rips~\cite{hellmerDensitySensitiveBifiltered2024,lesnickNerveModelsSubdivision2024}.
Most recently, Lesnick and
McCabe~\cite{lesnickNerveModelsSubdivision2024,lesnickSparseApproximationSubdivisionRips2024},
in the more particular case of metric spaces of bounded doubling dimension
(which include Euclidean spaces), give a $(1+\epsilon)$-approximation of
subdivision Rips whose $m$-skeleton has $O(\abs{X}^{m+2})$ simplices, for fixed
$\epsilon$ and dimension. Our methods also apply to the subdivision Rips
setting, yielding analogous results, as discussed in~\cref{sec:rips}.

\subparagraph{Miniball} In~\cref{sec:intersection_balls}, we compute the first
scale at which a set of sparse balls have non-empty intersection (in Euclidean space), which is required to
compute the sparse subdivision bifiltration. If instead of sparse balls we would
be using usual balls, such a scale would be the radius of the minimum enclosing
ball of the centers of the balls---the \textit{miniball problem}, first stated
by Sylvester in 1857~\cite{sylvesterQuestionGeometrySituation1857}. The miniball
problem can be solved in randomized linear time in the number of centers using
Welzl's algorithm~\cite{welzlSmallestEnclosingDisks1991}. In fact, the miniball
is an example of an \textit{LP-type} problem, as later defined by Matou\v{s}ek, Sharir and
Welzl~\cite{matousekSubexponentialBoundLinear1996}, and Welzl's algorithm for
the miniball can be generalized to the Matou\v{s}ek-Sharir-Welzl (MSW) algorithm
for LP-type problems. Fischer and
G{\"a}rtner~\cite{fischerSmallestEnclosingBall2004} solve the problem of computing
the minimum enclosing ball of balls---generalizing the miniball---via the
MSW-algorithm in a way that is practical and efficient, as in the implementation
in the CGAL library~\cite{cgal:fghhs-bv-24a}. We frame the problem of computing the first scale of
intersection of sparse balls as an LP-type problem and solve it as an extension
of Fischer and G{\"a}rtner's methods. In fact, this intersection problem can be
written as the \textit{smallest enclosing ball for a point set with strictly
  convex level sets}~\cite{zurcherSmallestEnclosingBall2007}, which is also solved as an
extension of Fischer and G{\"a}rtner's approach.

The miniball and similar problems have also been stated and solved
through the geometric optimization point of view,
see~\cite{dearingMinimumCoveringEuclidean2023,cawoodWeightedEuclideanOnecenter2024}
and references therein.

\subsection{Acknowledgements} The author would like to thank his advisor Michael
Kerber for helpful comments. He would also like to deeply thank an anonymous
reviewer for many comments that helped to substantially improve the exposition.
Part of the writing was carried out while the author was visiting Yasuaki Hiraoka's
group at the Kyoto University Institute for Advanced Study. This research was funded in
whole, or in part, by the Austrian Science Fund (FWF) 10.55776/P33765.

\section{A sparse multicover bifiltration}

In this section, we introduce the \deff{sparse multicover bifiltration}. As
already noted in the introduction, instead of using balls of radius $r$, the
sparse multicover involves balls with a lifetime of three phases: they start as
usual balls of radius $r$ at scale $r$, at some point they start to \textit{grow
  slowly}, that is, they have radii smaller than $r$ at scale $r$, and
eventually they disappear. We call this version of balls \textit{sparse balls}.
For their definition and the subsequent proofs, we use the notion of \emph{persistent
nets}, which we define first.

\subsection{Persistent nets}\label{sec:persistence_nets}

\deff{Persistent nets}, defined below, are a reinterpretation of
farthest point sampling, or greedy permutations, as known in the sparsification
literature~\cite{cavannaGeometricPerspectiveSparse2015,sheehySparseDelaunayFiltration2021,buchetSparseHigherOrder2024},
and as such are guaranteed to exist by Gonzalez's
algorithm~\cite{gonzalezClusteringMinimizeMaximum1985}. We will revisit these concepts in~\cref{sec:computation}.

Let $(X, \partial)$ be a finite metric space. We say that a subset
$S \subset X$ is an \deff{$r$-net} for a radius $r\geq 0$, if it is
\begin{enumerate}
  \item an $r$-covering: for every $x\in X\setminus S$ there is an $a\in S$ with
        $\partial(x, a) < r$, and
  \item an $r$-packing: for every two $a,b\in S$ we have $\partial(a, b) \geq r$.
\end{enumerate}
The condition of $r$-covering is equivalent to maximality, in the following
sense, which is why we take a strict inequality in the definition of
$r$-covering.
\begin{lemma}
  An $r$-packing $S\subset X$ is an $r$-covering if and only if it is maximal:
  there is no other $r$-packing $S'\subset X$ with $S\subset S'$.
\end{lemma}
\begin{proof}
  If $S$ is an $r$-packing and an $r$-covering, then
  $S\cup\Set{x}$ cannot be an $r$-packing for any $x\in X\setminus S$, because
  there is an $a\in S$ with $\partial(x, a) < r$---so $S$ is maximal.

  Suppose that $S\subset X$ is a maximal $r$-packing. For every
  $x\in X\setminus S$ there is an $a\in S$ with $\partial(x, a) < r$, because
  otherwise $S\cup\Set{x}$ would be an $r$-packing.
\end{proof}

In the spirit of persistence, a \deff{persistent net} $\mathcal{S}$ of $X$ is a collection of
$r$-nets $\mathcal{S}(r)\subset X$, one for each $r\geq 0$, such that for any
two $r'\geq r$ we have $\mathcal{S}(r')\subset \mathcal{S}(r)$.

The \deff{insertion radius} $\ins(x)\in\R\cup\Set{\infty}$ of a point $x\in X$, with
respect to a persistent net $\mathcal{S}$, is the supremum of those $r$ such
that $x\in\mathcal{S}(r)$. There is a large enough $r$ so
that $\mathcal{S}(r)$ is a single point (it has to be an $r$-packing), and so
there is only one point $x$ with $\ins(x) = \infty$.

\subsection{Sparse balls}

In what follows, we work over a fixed finite subset $X\subset\R^{d}$, with
$\R^{d}$ equipped with any norm. We also
fix an error parameter $\epsilon > 0$ and a persistent net $\mathcal{S}$ of $X$,
that we often drop from the notation. We define the \deff{slowing time}
$\slow(x)$ of a point $x\in X$ by
\begin{equation*}
  \slow(x)\coloneqq \frac{1+\epsilon}{\epsilon}\ins(x).
\end{equation*}

\begin{definition}\label{def:sparse_ball}
  We define the \deff{sparse ball} $\SB(x, r)$ of $x\in X$ at scale $r$ as:
  \begin{equation*}
    \SB(x, r) \coloneqq
    \begin{cases}
      B(x, \rho_{x}(r)) & r \leq (1+3\epsilon) \slow(x),\\
      \varnothing & r > (1+3\epsilon) \slow(x),
    \end{cases}
  \end{equation*}
  where $B(x, s)$ denotes the closed ball around $x$ of radius $s$, and
  $\rho_{x}\colon\R\to\R$ is a \deff{radius function}: any strictly increasing function
  such that
  \begin{itemize}
    \item on the interval $[0, \slow(x)]$, is $\rho_{x}(r) = r$, and
    \item on the interval $(\slow(x), (1+3\epsilon)\slow(x)]$, satisfies
          \begin{gather*}
            L_{x}(r) \leq \rho_{x}(r) \leq U_{x}(r), \text{ where}\\
            L_{x}(r) \coloneqq \frac{1}{1+3\epsilon} r + \frac{\epsilon}{1+\epsilon} \slow(x) \text{
              and
            } U_{x}(r) \coloneqq \frac{1}{3(1+\epsilon)} r +  \frac{2 + 3\epsilon}{3(1+\epsilon)}\slow(x).
          \end{gather*}
  \end{itemize}
\end{definition}
For example, the radius function of a point $x\in X$ can be
\begin{equation*}
  \rho_{x}(r) = \begin{cases}
    r, & r < \slow(x),\\
    U_{x}(r), & r \geq \slow(x),
  \end{cases}
\end{equation*}
but other options are useful for computation (\cref{sec:intersection_balls}).
In any case, the sparse ball of $x$ has radius $r$ up to scale $\slow(x)$, it
then starts to grow slowly with radius at most $U_{x}(r) < r$, until eventually
disappearing at scale $(1+3\epsilon)\slow(x)$. We say that the sparse ball
around $x\in X$ is \deff{slowed} at scale $r$ if $r > \slow(x)$.

\begin{figure}
  \centering
  \includegraphics{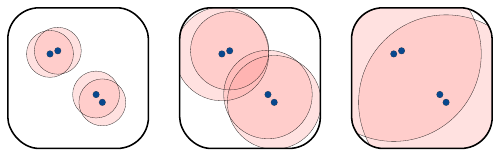}
  \caption{Sparse balls around four points in the plane. On the left square,
    they are not slowed yet. On the middle square, two of them are growing slowly,
    and are starting to get covered by the other two. On the right square, two
    of the sparse balls have disappeared.}\label{fig:sparse_balls}
\end{figure}

The functions $L_{x}(r)$ and $U_{x}(r)$ are defined in this way for technical reasons related
to properties of the sparse multicover that we will prove shortly. For now,
let us say that whenever a sparse ball disappears (that is, at scale $(1+3\epsilon)\slow(x)$),
it is covered by another non-slowed sparse ball:
\begin{lemma}
  Let $x\in X$ be a point and write $\gamma\coloneqq (1+3\epsilon)\slow(x)$.
  There exists a point $y\in X$ such that $\slow(y) \geq \gamma$ and
    $\SB(x, \gamma) \subset \SB(y, \gamma)$.
\end{lemma}
\begin{proof}
  The sparse ball of $x$ has radius at most
  $L_{x}(\gamma) = U_{x}(\gamma) = \frac{1+2\epsilon}{1+\epsilon} \slow(x)$.
At scale $\gamma$, the centers of the non-slowed sparse balls are precisely the
$\frac{\epsilon}{1+\epsilon}\gamma$-net
$\mathcal{S}(\frac{\epsilon}{1+\epsilon}\gamma)$. Thus, since
$\frac{\epsilon}{1+\epsilon}\gamma + \frac{1+2\epsilon}{1+\epsilon}\slow(x) < \gamma$,
there is a $y$ in $\mathcal{S}(\frac{\epsilon}{1+\epsilon}\gamma)$ such that the
sparse ball around $y$ covers the ball around $x$ of radius
$\frac{1+2\epsilon}{1+\epsilon}\slow(x)$,
\begin{equation*}
  \SB(x, \gamma) \subset B\left(x, \frac{1+2\epsilon}{1+\epsilon}\slow(x)\right) \subset B(y, \gamma) = \SB(y, \gamma).\qedhere
\end{equation*}
\end{proof}

\subsection{Covering map} One of the intuitive principles behind the sparse
multicover is that, at every scale $r\geq 0$, the non-empty sparse balls
\textit{approximately cover} the balls of radius $r$ around all the points of
$X$. Below, we formally define a map that takes each point $x$ to the non-empty
sparse ball that approximately covers the ball of $x$.

\begin{definition}\label{def:covering_sequence}
  We define the \deff{covering sequence} $(x_{0}, \dots, x_{m})$ of a point
  $x\in X$ as follows: we set $x_{0} = x$ and, recursively, given $x_{i-1}$ and
  if $\slow(x_{i-1}) < \infty$, we set $x_{i}$ to be the nearest neighbor
  (breaking ties arbitrarily) of $x$ among those points $y$ with
  $\slow(y)\geq (1+3\epsilon)\slow(x_{i-1})$. In other words, the nearest
  neighbor of $x$
  among those points whose sparse balls are non-slowed when the sparse ball of
  $x_{i-1}$ disappears.
\end{definition}

\begin{definition}\label{def:covering_map}
  For each scale $r\geq 0$, the \deff{covering map}
  $\mathcal{C}_{r}\colon X\to X$ takes $x$ to the first $x_{i}$ in its covering
  sequence whose sparse ball is non-empty at scale $r$.

  For a point $x\in X$ and $r\geq 0$, we say that
  $c_{r}(x)\coloneqq \abs{\mathcal{C}_{r}^{-1}(x)}$, the cardinality of its
  preimage under the covering map, is its \deff{covering weight}.
\end{definition}

We will use the following lemma multiple times. It describes how the covering
map inherits covering and packing properties from the persistent net.
\begin{lemma}\label{lem:covering_properties}
  The covering map has the following properties, for each $r\geq 0$:
  \begin{enumerate}
    \item The image $\mathcal{C}_{r}(X)$ is the set of points whose sparse
          balls at scale $r$ are non-empty,
    \item the points in $\mathcal{C}_{r}(X)$ are
          $\frac{\epsilon}{(1 + \epsilon)(1 + 3\epsilon)}r$-packed, and
    \item for each $x\in X$ and writing $y\coloneqq \Cont_{r}(x)$,
          $\norm{x - y} \leq \frac{\epsilon}{1 + \epsilon}\min(r, \slow(y))$.
  \end{enumerate}
\end{lemma}
\begin{proof}
  The first property is by definition.

  For the second, by the previous property the sparse ball of
  $x\in\Cont_{r}(X)$ has not disappeared at scale $r$, and thus
  $r \leq (1+3\epsilon)\slow(x) = (1 + 3\epsilon)\frac{1+\epsilon}{\epsilon}\ins(x)$.
  This gives that $\frac{\epsilon}{(1+\epsilon)(1 + 3\epsilon)}r \leq \ins(x)$,
  which implies what we want by the definition of $\ins(x)$ and persistent net.

  For the third, if $x = y$ there is nothing to prove. Otherwise, let $z$ be
  the point directly before $y$ in the covering sequence of $x$, and write
  $\gamma \coloneqq (1+3\epsilon)\slow(z)$. By construction, $\gamma < r$ and
  $\gamma \leq \slow(y) = \frac{1 + \epsilon}{\epsilon} \ins(y)$. Thus, $y$ is in
  $\mathcal{S}(\frac{\epsilon}{1+\epsilon}\gamma)$, and $y$ is the nearest neighbor
  of $x$ in such subset. We conclude that
  $\norm{x-y} \leq \frac{\epsilon}{1+\epsilon}\gamma \leq \frac{\epsilon}{1+\epsilon}\min(r,\slow(y))$.
\end{proof}

\subsection{The sparse multicover}
We can write the $k$-cover $\Cov(r,k)$ for an $r\geq 0$ as
  \begin{equation*}
    \Cov(r,k) \coloneqq \Set*{p\in\R^{d} \given \text{exists $S\subset X$ with $p\in\bigcap_{x\in S} B(x, r)$ and $\abs{S} \geq k$}}.
  \end{equation*}
In parallel, the sparse multicover uses sparse balls instead of balls, and the covering
weight instead of the number of balls:
\begin{definition}\label{def:sparse_multicover}
  The \deff{sparse $k$-cover} $\SCov(r,k)$, for $r \geq 0$ and $k\in\N$, is
  the set of points $p\in\R^{d}$ covered by sparse balls at scale $r$
  whose sum of covering weights at scale $r$ is at least $k$,
  \begin{equation*}
    \SCov(r,k) \coloneqq \Set*{p\in\R^{d} \given \text{exists $S\subset X$ with $p\in\bigcap_{x\in S} \SB(x, r)$ and $\sum_{x\in S} c_{r}(x) \geq k$}}.
  \end{equation*}
\end{definition}

  The \deff{sparse multicover} is the bifiltration given by the sparse $k$-covers:

\begin{proposition}\label{prop:filtration}
  For $r\leq r'$ and $k\geq k'$, we have $\SCov(r,k) \subset \SCov(r',k')$.
\end{proposition}
\begin{proof}
  Let $p\in\SCov(r,k)$ and let $S$ be a subset whose sparse balls at scale $r$
  cover $p$ with covering weight at least $k$. Take
  $A \coloneqq \Cont_{r}^{-1}(S)$ and note that there are at least $k$ points in
  $A$. For each $a\in A$, we claim that the sparse ball of $\Cont_{r'}(a)$ (the point that ``covers'' $a$ at scale $r'$) also contains $p$.
  Proving the claim finishes the proof, as the sum of covering weights of
  $\Cont_{r'}(A)$ is at least $k \geq k'$.

  We write $x\coloneqq \Cont_{r}(a)$ and $y\coloneqq\Cont_{r'}(a)$. The claim is
  immediate if $x = y$. Otherwise, note that necessarily $\ins(x) < \infty$,
  because the point with infinite insertion radius is unique. Also note that
  $x\not\in\Cont_{r'}(X)$ and, thus,
  $(1+3\epsilon)\slow(x) \leq \min(r', \slow(y))$. This fact, together with the triangle inequality
  and~\cref{lem:covering_properties}, gives
  \begin{equation*}
    \begin{split}
      \norm{y - p} &\leq \norm{y - a} + \norm{a - x} + \norm{x - p} \\
                   &\leq \frac{\epsilon}{1+\epsilon} \min(r', \slow(y)) + \frac{\epsilon}{1+\epsilon} \min(r, \slow(x)) + \rho_{x}(r).\\
                   &\leq \frac{\epsilon}{1+\epsilon} \min(r', \slow(y)) + \frac{\epsilon}{1+\epsilon} \slow(x) + \frac{1+2\epsilon}{1+\epsilon}\slow(x)\\
                   &\leq \frac{\epsilon}{1+\epsilon} \min(r', \slow(y)) + \frac{1 + 3\epsilon}{1+\epsilon} \slow(x)\\
                   &\leq \left(\frac{\epsilon}{1+\epsilon} + \frac{1}{1+\epsilon}\right) \min(r', \slow(y))\\
                   &= \min(r',\slow(y)) \leq \rho_{y}(r').\qedhere
    \end{split}
  \end{equation*}
\end{proof}

The sparse multicover approximates the multicover:
\begin{theorem}\label{thm:approximation}
  Writing
  $\delta\coloneqq\frac{1+2\epsilon}{1+\epsilon} < 1+\epsilon$, we have
  \begin{equation*}
    \Cov(r,k) \subset \SCov((1+3\epsilon)r,k) \text{ and } \SCov(r,k)\subset \Cov(\delta r,k).
  \end{equation*}
\end{theorem}
\begin{proof}
  We start with the first inclusion. Let $p\in\Cov(r,k)$, which implies that
  there is a subset $S\subset X$ of $k$ points such that $\norm{a - p} \leq r$
  for all $a\in S$.

  Let $S' = \Cont_{(1+3\epsilon)r}(S)$ and note that the sum of covering weights
  $\sum_{x\in S'}c_{(1+3\epsilon)r}(x)$ is at least $k$. It is left to show that every sparse ball
  around points of $S'$ at scale $(1+3\epsilon)r$ contains $p$. Consider a point
  $a\in S$ and write $x\coloneqq \Cont_{(1+3\epsilon)r}(a)\in S'$. Note that
  necessarily $r \leq \slow(x)$, because otherwise the sparse ball of $x$ has
  disappeared by scale $(1+3\epsilon)r$. We have
  \begin{equation*}
    \norm{x - p} \leq \norm{x - a} + \norm{a - p} \leq \frac{\epsilon}{1+\epsilon}\min((1+3\epsilon)r, \slow(x)) + r.
  \end{equation*}
  Now we do a case distinction. If $(1+3\epsilon)r < \slow(x)$, we have
  \begin{equation*}
    \norm{x - p} \leq \frac{\epsilon}{1+\epsilon}(1+3\epsilon)r + r \leq (1+3\epsilon) r =  \rho_{x}((1+3\epsilon)r),
  \end{equation*}
  and otherwise $(1+3\epsilon)r \geq \slow(x)$, so we have
  \begin{equation*}
    \norm{x - p} \leq \frac{\epsilon}{1+\epsilon}\slow(x) + r = L_{x}((1+3\epsilon)r) \leq \rho_{x}((1+3\epsilon) r).
  \end{equation*}

  We now prove the second inclusion. Let $p\in \SCov(r,k)$. This means that
  there exists a subset $S\subset \mathcal{C}_{r}(X)$ whose sparse balls cover
  $p$ at scale $r$ with covering weight at least $k$. Consider
  $A \coloneqq \mathcal{C}_{r}^{-1}(S)$. Then, by definition, $\abs{A} \geq k$
  and it is left to show that for each $a\in A$ we have
  $\norm{a-p}\leq \delta r$. Indeed, writing $x\coloneqq\Cont_{r}(a)$,
  \begin{equation*}
    \norm{a - p} \leq \norm{a - x} + \norm{x - p} \leq \frac{\epsilon}{1+\epsilon}\min(r, \slow(x)) + \rho_{x}(r).
  \end{equation*}
  We do a case distinction again. If $r\leq\slow(x)$, then
  \begin{equation*}
    \norm{a-p}\leq \frac{\epsilon}{1+\epsilon}r + r = \frac{1+2\epsilon}{1+\epsilon}r.
  \end{equation*}
  Otherwise, if $r > \slow(x)$,
  \begin{equation*}
    \begin{split}
      \norm{a-p} &\leq \frac{\epsilon}{1+\epsilon}\slow(x) + U_{x}(r)\\
                   &= \frac{\epsilon}{1+\epsilon}\slow(x)
      + \frac{1}{3(1+\epsilon)}r
      + \frac{2 + 3\epsilon}{3(1+\epsilon)}\slow(x)
      \leq \frac{1+2\epsilon}{1+\epsilon}r.\qedhere
    \end{split}
  \end{equation*}
\end{proof}

\section{An equivalent simplicial bifiltration}

We look for a bifiltration of simplicial complexes that is homotopically
equivalent, as defined shortly, to the sparse multicover bifiltration. This will
be the \deff{sparse subdivision bifiltration}, which is an extension of the
subdivision \v{C}ech bifiltration of
Sheehy~\cite{sheehyMulticoverNerveGeometric2012}.

\subsection{Homotopically equivalent bifiltrations}\label{sec:topological_equivalence}
We now recall the fundamental definitions of the homotopy theory of filtrations, where
filtrations are viewed as diagrams, as is standard; see, e.g.,~\cite{blumbergStability2ParameterPersistent2024}. Let $F$ be a bifiltration of topological spaces
indexed, in our case, by radii $r\in [0, \infty)$ and order $k\in\N$; that is, a
collection of topological spaces $F_{r,k}$, one for each $r\geq 0$ and $k\in\N$,
together with a continuous map $F_{(r,k)\to (r',k')}\colon F_{r,k} \to F_{r',k'}$ for every $r\leq r'$ and
$k\geq k'$, such that $F_{(r,k)\to (r,k)}$ is the identity and the following diagram commutes
\begin{equation*}
  \begin{tikzcd}
    F_{r, k} \arrow[rr]\arrow[dr] & & F_{r',k'}\\
    & F_{r'',k''}. \arrow[ur] &
  \end{tikzcd}
\end{equation*}
A \deff{pointwise weak equivalence}
$F\xrightarrow{\isomorphic} G$ between two bifiltrations $F$ and $G$, is a
collection of weak homotopy equivalences $\alpha_{r,k}\colon F_{r,k}\to G_{r,k}$ such that the following diagram commutes for every
$r\leq r'$ and $k\geq k'$:
\begin{equation*}
  \begin{tikzcd}
    F_{r,k} \arrow[r]\dar[swap]{\alpha_{r,k}} & F_{r',k'}\dar{\alpha_{r',k'}}\\
    G_{r,k}  \arrow[r] & G_{r',k'}\nospaceperiod
  \end{tikzcd}
\end{equation*}
Pointwise weak equivalence is not an equivalence relation, but the following
is. Two bifiltrations $F$ and $G$ are \deff{weakly equivalent} if there
exists a zigzag of pointwise homotopy equivalences that connect them:
\[
  \begin{tikzcd}[ampersand replacement=\&,column sep=2ex,row sep=2ex]
    \& C^1\ar[swap]{dl}\ar{dr}  \&           \& \cdots\ar[swap]{dl}\ar[]{dr}  \&                \&   C^n\ar[swap]{dl}\ar[]{dr}  \\
    F \&                                                             \&  C^2 \&                                                                                                                       \& C^{n-1} \&                                                               \&G.
  \end{tikzcd}
\]

\subsection{The sparse subdivision bifiltration}
We now define the \deff{sparse subdivision bifiltration} and prove that it is
weakly equivalent to the multicover bifiltration. As already mentioned, the
sparse subdivision bifiltration is an extension of Sheehy's \deff{subdivision
  \v{C}ech bifiltration}~\cite{sheehyMulticoverNerveGeometric2012}, which we review
first.

\subparagraph{Subdivision \v{C}ech bifiltration}
We start with the following extension of the \v{C}ech complex:
\begin{definition}\label{def:cech_poset}
  For a given scale $r\in[0, \infty)$ and order $k\in\N$, the \deff{$k$-\v{C}ech
  poset} $\Cech(r,k)$ is
  \begin{equation*}\label{eq:cech_poset}
    \Cech(r,k) \coloneqq \Set*{\sigma\subset X \given \bigcap_{x\in\sigma} B(x, r) \neq \varnothing \text{ and } \abs{\sigma} \geq k}
  \end{equation*}
  with the order given by inclusion.
\end{definition}
The $1$-\v{C}ech poset is (the face poset of) the usual \v{C}ech simplicial
complex. The collection of the \v{C}ech posets $\Cech(r,k)$ assembles into a
bifiltration that we call the \deff{multi-\v{C}ech bifiltration} $\Cech$: for any $r\leq r'$
and $k \geq k'$, we have $\Cech(r,k)\subset \Cech(r',k')$.

To obtain a bifiltration of simplicial complexes we take the order complex
pointwise, as below. By \deff{order complex} $\order(P)$ of a poset $P$ we mean the
simplicial complex whose set of vertices is $P$ and whose $m$-simplices are the
chains of $P$ of length $m+1$.

\begin{definition}
  The \deff{subdivision \v{C}ech bifiltration} $\Sub$ is given by
  $\Sub(r,k) \coloneqq \order(\Cech(r,k))$ with the internal maps being the inclusions.
\end{definition}

Now we can state the multicover nerve theorem~\cite{sheehyMulticoverNerveGeometric2012,cavannaWhenWhyTopological2017,blumbergStability2ParameterPersistent2024}:

\begin{restatable}[Multicover nerve theorem]{theorem}{firstmulticovernerve}\label{thm:multicover_nerve}
  The multicover bifiltration $\Cov$ is weakly equivalent to the subdivision \v{C}ech
  bifiltration $\Sub$.
\end{restatable}

\subparagraph{Sparsification} We now extend the subdivision \v{C}ech bifiltration and the
multicover nerve theorem to the
sparse multicover bifiltration. As in the definition of the sparse multicover,
we replace the balls with sparse balls and the cardinality with the covering
weight, resulting in the following definition; to be compared with the
definition of $k$-\v{C}ech poset (\cref{def:cech_poset}).

\begin{definition}\label{def:sparse_intersection_poset}
  For a given scale $r \in [0,\infty)$ and order $k\in\N$, we define the \deff{sparse
    \v{C}ech poset} $\SCech(r,k)$ to be
  \begin{equation*}
    \begin{split}
      \SCech(r,k) &\coloneqq \bigcup_{s\leq r}\Set*{\sigma\subset X \given \bigcap_{x\in\sigma}\SB(x, s)\neq\varnothing \text{ and } \sum_{x\in\sigma}c_{s}(x) \geq k}%
    \end{split}
  \end{equation*}
  with the order given by inclusion.
\end{definition}
In the definition, the union is there to guarantee that the $\SCech(r,k)$
assemble into a bifiltration, because now sparse balls disappear, unlike
before; this is the combinatorial analogue of taking \textit{cones} as in
Cavanna, Jahanseir and Sheehy's work~\cite{cavannaGeometricPerspectiveSparse2015}. Again taking the order complex:

\begin{definition}\label{def:sparse_subdivision}
  The \deff{sparse subdivision bifiltration} $\SSub$ is given by
  $\SSub(r,k)\coloneqq \order(\SCech(r,k))$, the order complex of the sparse
  \v{C}ech poset $\SCech(r,k)$.
\end{definition}
In other words, the $m$-simplices of $\SSub(r,k)$ are sequences
$\sigma_{0} \subset \dots \subset \sigma_{m}$ of $m + 1$ subsets of points, all
in $\SCech(r,k)$ and where each containment is strict.

\begin{restatable}[Sparse multicover nerve theorem]{theorem}{firstsparsenerve}\label{thm:sparse_multicover_nerve}
  The sparse subdivision bifiltration $\SSub$ is weakly equivalent to the sparse
  multicover $\SCov$.
\end{restatable}
The proof is deferred to~\cref{sec:proof_sparse_nerve}. It reduces to the multicover nerve
theorem, incorporating the covering weight and using the \textit{cones} strategy
of Cavanna, Jahanseir and Sheehy~\cite{cavannaGeometricPerspectiveSparse2015}.

\section{Size}\label{sec:size}

We are now ready to bound the size of the sparse subdivision
bifiltration. Specifically, we bound the number of simplices in the largest
complex in the bifiltration, and show that each simplex is born at a constant number of
\deff{critical grades}. We say that the grade $(r, k)\in[0, \infty)\times\N$ is
\deff{critical} for a simplex $\sigma$, if $\sigma\in\SSub(r',k')$ for all
$r < r'$ and $k \geq k'$, but
$\sigma\not\in\SSub(s,l)$ for any $s < r$ and $l > k$.

\subsection{Number of simplices}

The rest of this section is dedicated to proving the following theorem. We write
$\SSub$ to refer both to the bifiltration and to the largest complex in the bifiltration,
$\SSub = \bigcup_{r,k}\SSub(r,k)$, when no confusion is possible.

\begin{theorem}\label{thm:size_simplices}
  The bifiltration $\SSub$ has $O(\abs{X})$ simplices for fixed
  $\epsilon$ and $d$.
\end{theorem}
The following argument is analogous to the one used in previous sparse
filtrations~\cite{sheehyLinearSizeApproximationsVietoris2013,cavannaGeometricPerspectiveSparse2015}.
Consider the map $\min\colon\SSub\to X$ that takes each simplex
$\sigma = \sigma_{0}\subset\cdots\subset\sigma_{m}$ to the minimum
point $x$ in $\sigma_{m}$ with respect to the order given by their insertion radius
$\ins(x)$, breaking ties arbitrarily. The points in each preimage of this map
$\min\colon\SSub\to X$ are not too far from each other:
\begin{lemma}\label{lem:covering_of_friends}
  Let $\sigma = \sigma_{0}\subset\cdots\subset\sigma_{m}\in\SSub$ be a simplex
  and let $x\coloneqq\min(\sigma)$. The points $\sigma_{m}$ are contained in a
  ball of radius $2(1+3\epsilon)\slow(x)$ around $x$.
\end{lemma}
\begin{proof}
  Since the sparse ball around $x$ is empty at scales greater than
  $\alpha\coloneqq (1+3\epsilon)\slow(x)$, any intersection between the sparse
  balls of $\sigma_{m}$ must happen at a scale $r\leq \alpha$. Let $p$ be any
  point in this intersection. It follows that for any other $y\in\sigma_{m}$ we
  have $\norm{x - y} \leq \norm{x - p} + \norm{p - y} \leq 2\alpha$.
\end{proof}

\begin{lemma}\label{lem:simplex_min_bound}
  For each $x\in X$, the number of simplices mapped to $x$ under $\min$ is
  $O(1)$, for fixed $\varepsilon$ and dimension $d$.
\end{lemma}
\begin{proof}
  Let
  $A\coloneqq \bigcup_{\sigma_{0}\subset\cdots\subset\sigma_{m}\in\min^{-1}(x)}\sigma_{m}$
  be the set of all points in simplices $\sigma$ such that $\min\sigma = x$.
  It suffices to show that $A$ has $O(1)$ points, for fixed
  $\varepsilon$ and $d$.
  Let $\alpha \coloneqq (1+3\epsilon)\slow(x)$. We claim that
  $A\subset\Cont_{\alpha}(X)$. Indeed, for every $a\in A$ we have
  $\ins(a) \geq \ins(x)$, by the definition of $\min$, and thus the sparse ball
  of $a$ is non-empty at scale $\alpha$.

  By~\cref{lem:covering_of_friends}, the points in $A$ are contained in a ball
  of radius $2\alpha$, and, by~\cref{lem:covering_properties}, they are
  $\frac{\epsilon}{(1 + \epsilon)(1+3\epsilon)}\alpha$-packed, because
  $A\subset\Cont_{\alpha}(X)$ as shown. By comparing the volume of balls of
  radius $2\alpha$ and $\frac{\epsilon}{(1 + \epsilon)(1+3\epsilon)}\alpha$, we
  conclude that $A$ consists of
  $O\left(\left(\frac{2(1+\epsilon)(1+3\epsilon)}{\epsilon}\right)^{d}\right)$
  points.
\end{proof}

All in all, there can be at most $O(\abs{X})$ simplices, since at most a constant
number, for a fixed $\epsilon$ and $d$, of simplices is mapped to each point,
finishing the proof of the theorem.

\subsection{Critical grades} The sparse subdivision bifiltration, unlike the
subdivision \v{C}ech bifiltration, is
\textit{multicritical}, which means that each simplex may have multiple critical
grades. Let us briefly describe the critical grades of a simplex
$\sigma = \sigma_{0}\subset\cdots\subset\sigma_{m}\in\SSub$. Ordering the
critical grades by scale $r$, the first critical grade has scale $r$ equal to
the first scale at which the sparse balls around the points $\sigma_{m}$
intersect, and order $k$ equal to the sum of the covering weights of the points
of $\sigma_{0}$ at such a scale. Further critical grades of $\sigma$ correspond to scales $r'$ at which
the covering weight of the points of $\sigma_{0}$ increases, until one of the
associated sparse balls disappears. Still, there are not many
critical grades, as another packing argument shows:

\begin{theorem}\label{thm:size_critical_grades}
  Each simplex in $\SSub$ has $O(1)$ critical grades, for fixed $d$ and
  $\epsilon$.
\end{theorem}
\begin{proof}
  Let $\sigma = \sigma_{0}\subset\cdots\subset\sigma_{m}\in\SSub$ and
  let $x\coloneqq \min\sigma$. The simplex $\sigma$ has no critical grade with
  scale greater than $\alpha\coloneqq (1+3\epsilon)\slow(x)$, because the sparse
  ball of $x$ disappears then. Thus, we bound the number of scales
  $s\leq \alpha$ at which the sum of covering weights
  $\sum_{x\in \sigma_{0}} c_{s}(x)$ of the points of $\sigma_{0}$ increases,
  which bounds the number of critical grades of $\sigma$.

  By the proof of~\cref{lem:simplex_min_bound}, the points in $\sigma_{0}$
  are a $2\beta$-packing, where
  $\beta \coloneqq \frac{\epsilon}{2(1+\epsilon)(1+3\epsilon)}\alpha$. Thus, any
  intersection between sparse balls around points in $\sigma_{0}$ must happen at scales
  at least $\beta$, and thus we are interested in the range of scales
  $[\beta,\alpha]$.

  Let $A\subset X$ be the subset of points $a\in X\setminus \sigma_{0}$ whose
  sparse ball disappears in the interval $[\beta, \alpha)$, meaning
  $(1+3\epsilon)\slow(a)\in[\beta, \alpha)$ and, as a result, the covering
  weight of a point in $\sigma_{0}$ increases, meaning there is a point $y\in X$
  such that $a$ and a point $w\in\sigma_{0}$ are one after the other in the
  covering sequence of $y$; that is, $y_{i} = a$ and $y_{i+1} = w$.
  To finish the argument, it suffices to show that $\abs{A} = O(1)$. First, note that
  $A\subset\Cont_{\beta}(X)$. It follows that the points in $A$ are
  $\frac{\epsilon}{(1+\epsilon)(1+3\epsilon)}\beta$-packed,
  by~\cref{lem:covering_properties}.

  Let $a\in A$, and let $y\in X$ and $w\in\sigma_{0}$ as in the definition of $A$. By~\cref{lem:covering_properties} and~\cref{lem:covering_of_friends},
  \begin{equation*}
    \begin{split}
      \norm{x - a} &\leq \norm{x - w} + \norm{w - y} + \norm{y - a}\\
                   &\leq 2\alpha + (1+3\epsilon)\slow(a) + \frac{\epsilon}{1+\epsilon}\slow(a) \leq 4\alpha.
    \end{split}
  \end{equation*}
  All in all, the points of $A$ are in a ball of radius $4\alpha$ and
  are all
  $\frac{\epsilon}{(1+\epsilon)(1+3\epsilon)}\beta=\frac{\epsilon^{2}}{2(1+\epsilon)^{2}(1+3\epsilon)^{2}}\alpha$
  apart. The result follows by comparing the volume of disjoint balls around
  points in $A$ and the volume of a ball of radius $4\alpha$.
\end{proof}

\section{Computation}\label{sec:computation}

We show how to compute the sparse subdivision bifiltration. The first step is
computing a persistent net and the covering map. As already mentioned
in~\cref{sec:persistence_nets} this reduces to computing a \textit{greedy
  permutation}. The next step is to compute all potential simplices that appear
in the bifiltration, as those consisting of points around balls like those
of~\cref{lem:covering_of_friends}; for which we use a suitable proximity data
structure. Finally, to obtain the critical grades of the simplices we can use
the covering map and the minimum scale at which a subset of the sparse balls
intersect. The problem of computing this minimum scale, in the Euclidean case, can be written as an
LP-type problem~\cite{matousekSubexponentialBoundLinear1996}, and we show how to
solve it, much like computing the smallest enclosing ball of
balls~\cite{fischerSmallestEnclosingBall2004}. All these steps can be done in
$O(\abs{X}\log\Delta)$ time, where $\Delta$ is the \textit{spread} of $X$, the ratio
between the longest and shortest distances in $X$.

\subsection{Greedy permutations}
Let $(x_{1}, \dots, x_{n})$ be an ordering of
the $n$ points in $X$. Such a sequence is called \deff{greedy} if for every
$x_{i+1}$ and prefix $X_{i}\coloneqq \Set{x_{1}, \dots, x_{i}}$, one has
\begin{equation*}
  d(x_{i+1}, X_{i}) = \max_{x\in X} d(x, X_{i}),
\end{equation*}
where $d(x, X_{i}) = \min_{x_{j}\in X_{i}} \norm{x - x_{j}}$. The
distance $r_{i+1}\coloneqq d(x_{i+1}, X_{i})$ is the \deff{insertion radius}
of $x_{i+1}$, where we take $r_{1} = \infty$ by convention. A
greedy permutation gives a persistent net $\mathcal{S}$ by taking
  $\mathcal{S}_{r} \coloneqq \Set{x_{i} \given r \leq r_{i}}$.

Gonzalez~\cite{gonzalezClusteringMinimizeMaximum1985} gives a method to compute
a greedy permutation. It consists of $n$ phases, divided into an initialization
phase and $\abs{X}-1$ update phases. At the end of phase $i$, it maintains a subset
 $S_{i} = \Set{x_{1}, \dots, x_{i}}\subset X$ and for
each $x_{j}\in S_{i}$ its \deff{Voronoi set}: the subset of points in $X$ that are
closer to $x_{j}$ than to any other point in $S_{i}$. We break ties arbitrarily
but consistently (say, by minimum index in $S_{i}$), so that the Voronoi sets
form a partition of $X$. In the initialization
phase we pick a random point $x_{1}$ and we set its Voronoi set to be every
other point. In the update phase $i$, we obtain the new set of leaders $S_{i}$ by adding to $S_{i-1}$ the point $x_{i}$ whose distance to $S_{i-1}$ is maximal, and updating
the Voronoi sets. This algorithm can be implemented in $O(\abs{X}^{2})$-time by going
over the whole set of points during each phase. An algorithm of
Clarkson~\cite{clarksonFastAlgorithmsAll1983,clarksonNearestNeighborSearching2003}, which also maintains the Voronoi sets,
can be shown to take
$O(\abs{X}\log\Delta)$-time~\cite{har-peledFastConstructionNets2006a}, and has been
implemented~\cite{clarksonNearestNeighborSearching2003,sheehyGreedypermutations}.

As mentioned, at every phase of the algorithm, each point $x\in X$ is assigned a Voronoi
set of a point $x_{j}\in S_{i}$: its nearest neighbor among those in $S_{i}$. We call the
sequence of such points $x_{j}$, as the algorithm progresses, the \deff{sequence of
  leaders} of a point $x$. Such a sequence necessarily starts with $x_{1}$ and ends
with the point itself.

From the sequence of leaders we can obtain the covering
map. The covering sequence of a point is a subsequence of its
sequence of leaders, in reverse order. To see this, note that the covering
sequence of $x$ starts with $x$ itself, as in the reversed sequence of leaders.
Then, when the sparse ball of $x$ disappears the next point in its covering
sequence is its nearest neighbor among those points with non-slowed sparse ball.
This set of points with non-slowed sparse ball is necessarily equal to one of
the subsets $S_{i}$ we have encountered through the execution of the
algorithm. In addition, balls are slowed in the same order as the reversed
sequence of leaders. All in all, we conclude the following:

\begin{lemma}
  A persistent net, the insertion radius of all the points, and the covering map
  can be computed in time $O(\abs{X}\log\Delta)$.
\end{lemma}

\subsection{Neighborhoods}
We also need to find those subsets of points whose
sparse balls have a non-empty intersection, at any scale.
By~\cref{lem:covering_of_friends}, any such subset $A$ is contained in a ball of
radius $2(1+3\epsilon)\slow(x)$ around $x$, where $x$ is the point of minimal
insertion radius in $A$. Making use of this observation, and
following~\cite{buchetSparseHigherOrder2024}, given a point $x\in X$, we call
the points of higher insertion radius than the one of $x$ and within distance
$2(1+3\epsilon)\slow(x)$ its \deff{friends}. Note that the number of friends of
a point is $O(1)$, because we are only considering points of higher insertion
radius, as
in~\cref{lem:simplex_min_bound}. Once we have computed the friends of $x$, we
can go over every possible subset and check whether the associated sparse balls
intersect at some scale, which is done in the following section. This procedure
is similar to other sparsification
schemes~\cite{cavannaGeometricPerspectiveSparse2015,buchetSparseHigherOrder2024}.

To compute the friends of every point, we can use a proximity search data
structure as follows. First, we initialize it
empty. In decreasing order of insertion radius (that is, the order given by the
greedy permutation), we add the points one by one. After adding a point $x$, we
query all the points within $2(1+3\epsilon)\slow(x)$ distance from $x$.

Cavanna, Jahanseir and Sheehy~\cite{cavannaGeometricPerspectiveSparse2015} give
such a data structure to compute all friends in such a way in time $O(\abs{X})$, again for
fixed $\epsilon$ and $d$. In practice one can also use the related greedy
trees~\cite{chubetProximitySearchGreedy2023}, which have been
implemented~\cite{sheehyGreedypermutations}.

\begin{lemma}
  Computing the friends of all points can be done in $O(\abs{X})$ time.
\end{lemma}

\subsection{Intersection of sparse balls}\label{sec:intersection_balls}
We now show how to compute the scale at which a subset of sparse balls first
intersect, if they do so. We do it only for the Euclidean case; the
$l^{\infty}$-norm, which is of interest in~\cref{sec:rips}, is an easier case,
since it is enough to compute the intersection times pairwise.

Let us choose a specific radius function first, with the added restriction that
$0<\epsilon\leq 1$. For a point $x\in X$, we take
\begin{equation}\label{eq:quadratic_radius}
  \rho_{x}(r) =
  \begin{cases}
    r & r \leq \slow(x),\\
    \sqrt{K_{\epsilon} r^{2} + (1-K_{\epsilon})\slow(x)^{2}} & \slow(x) \leq r,
  \end{cases}
\end{equation}
where
  $K_{\epsilon} = \frac{1}{3(1+\epsilon)^{2}}$.
See~\cref{fig:graph}. It can be checked that $\rho_{x}(r)$ is indeed a radius
function.
\begin{figure}
  \centering
  \includegraphics{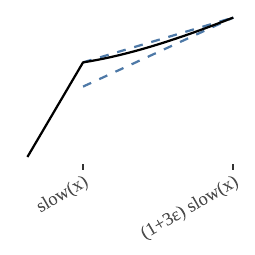}
  \caption{Graph of $\rho_x(r)$ of~\eqref{eq:quadratic_radius}, for some values
    $\slow(x)$ and $0 <\epsilon < 1$. The dashed lines are the functions
    $L_{x}(r)$ and $U_{x}(r)$, as in~\cref{def:sparse_ball}.%
  }\label{fig:graph}
\end{figure}

\subparagraph{Solving a problem}
Equipped with the radius function above and given a subset of points $A\subset X$,
we consider the problem
\begin{alignat}{2}\label{eq:problem_sparse_balls}
  \min_{z\in\R^{d},r\geq 0} \quad & r^{2}, & \\
  \textrm{subject to} \quad & \norm{a - z}^{2} \leq r^{2}, &\quad  a\in A,\nonumber\\
                                   & \norm{a - z}^{2} \leq K_{\epsilon}r^{2} + (1-K_{\epsilon})\slow(a)^{2},
                                     & a\in A.\nonumber
\end{alignat}
Now it is apparent why we chose the specific radius function
of~\cref{eq:quadratic_radius}: the right side of the inequalities above is
linear in $r^{2}$.
We claim that the solution of this problem gives the minimum scale at which the
sparse balls of $A$ have non-empty intersection, if such a scale exists. Let
$z_{\opt}\in\R^{d}$ and $r_{\opt}\geq 0$ be the solution of the problem
above (a unique solution always exists, as we show shortly). One can
check that $r_{\opt}$ is the minimum scale at which the balls
$B(a, \rho_{a}(r))$ around points $a\in X$ first intersect, by noting that
$r^{2} \leq K_{\epsilon}r^{2} + (1-K_{\epsilon})\slow(x)^{2}$ whenever $r\leq \slow(x)$. Thus, $r_{\opt}$ is
the minimum scale at which the sparse balls of $A$ intersect, unless a sparse
ball disappears before then, that is, unless
$r_{\opt} > (1+3\epsilon)\slow(a)$ for any $a\in A$.

We solve the problem of~\cref{eq:problem_sparse_balls} via a more general
problem. Given a set $H$ of $m$ constraints, each being a triple consisting of a
point $p_{i}\in\R^{d}$, a positive real constant $\alpha_{i}\in\R$ and a
non-negative real constant $\beta_{i}\in\R$, we look for
\begin{align}
  \min_{z\in\R^{d},s\geq 0} \quad & s, & \tag{M}\label{eq:problem_min} \\
  \textrm{subject to} \quad & \norm{p_{i} - z}^{2} \leq \alpha_{i}s + \beta_{i}, \quad i = 1, \dots, m. \nonumber
\end{align}
The problem of~\cref{eq:problem_sparse_balls} is an instance
of~\eqref{eq:problem_min} with $2\abs{A}$ constraints.

Note that by taking a set of points $p_{1}, \dots, p_{m}$ and setting
$\alpha_{i} = 1$ and $\beta_{i} = 0$, the solution to~\eqref{eq:problem_min}
above is the smallest enclosing ball of the points (with radius $\sqrt{s}$).

As explained in~\cref{sec:lp_type}, problem~\eqref{eq:problem_min} can be solved
using the framework of \textit{LP-type problems} and the MSW
algorithm~\cite{matousekSubexponentialBoundLinear1996}, following the strategy
for the smallest enclosing ball of balls of Fischer and
G{\"a}rtner~\cite{fischerSmallestEnclosingBall2004} (which is practical and
implemented in CGAL~\cite{cgal:fghhs-bv-24a}). In fact, problem~\eqref{eq:problem_min} fits the
\textit{smallest ball of point sets with strictly convex level sets}~\cite{zurcherSmallestEnclosingBall2007}, itself an extension of Fischer and
G{\"a}rtner's methods.

\section{The sparse multicover nerve theorem}\label{sec:proof_sparse_nerve}

In this section we prove the sparse multicover nerve
theorem,~\cref{thm:sparse_multicover_nerve}. The fundamental ingredients are the
nerve theorem, which we review first, a lemma that uses Quillen's theorem A and
draws from the proof of the multicover nerve theorem
in~\cite{blumbergStability2ParameterPersistent2024}
and~\cite{sheehyMulticoverNerveGeometric2012}, and a cover based on
\textit{cones} as in the work of Cavanna, Jahanseir, and
Sheehy~\cite{cavannaGeometricPerspectiveSparse2015}.

\subsection{The nerve theorem}
We use a version of the \textit{functorial} nerve theorem,
see~\cite{bauerUnifiedViewFunctorial2023}. Recall that given a cover
$A = (A_{i})_{i\in I}$, indexed by a finite set $I$, of a topological space,
often tacitly assumed to be $\bigcup_{i\in I} A_{i}$, the \deff{nerve} of $A$ is the
simplicial complex with vertex set $I$ and simplices
\begin{equation*}
  N(A) \coloneqq \Set*{\sigma\subset I \given \bigcap_{i\in\sigma} A_{i} \neq \varnothing}.
\end{equation*}

\begin{theorem}[Nerve theorem {\cite[Theorem 3.9]{bauerUnifiedViewFunctorial2023}}]\label{thm:nerve}
  Let $A = (A_{i})_{i\in I}$ be a finite cover of a $X\subset\R^{d}$, where
  each $A_{i}$ is closed and convex. Then there is a space $B(A)$ and homotopy
  equivalences $B(A)\to X$ and $B(A)\to N(A)$.

  Moreover, for any two $X\subset X'$ in $\R^{d}$, with closed and
  convex finite covers $(A_{i})_{i\in I}$ and $(A'_{j})_{j\in J}$, respectively, with $I\subset J$ and
  $A_{i}\subset A'_{i}$ for every $i\in I$, these homotopy equivalences assemble
  into a commutative diagram
\begin{equation*}
  \begin{tikzcd}
    X \ar[hook]{d} & \lar B(A) \rar \ar[hook]{d} & N(A) \ar[hook]{d}\\
    X' & \lar B(A') \rar & N(A'),
  \end{tikzcd}
\end{equation*}
where the vertical maps are inclusions.
\end{theorem}

In fact, the above holds in greater generality~\cite[Corollary 5.16, Theorem
5.9]{bauerUnifiedViewFunctorial2023}, as below. For a subset $\sigma\subset I$,
we denote by $A_{\sigma}$ the subset
\begin{equation*}
  A_{\sigma} = \bigcap_{i\in\sigma} A_{i}.
\end{equation*}

\begin{theorem}\label{thm:general_nerve}
  In~\cref{thm:nerve}, we can replace the condition of each $A_{i}$ being
  convex by the following: for each $\sigma\subset I$ the subset $A_{\sigma}$ is
  either empty or contractible, and for each $\sigma\subset\sigma'\subset I$ the
  pair $(A_{\sigma}, A_{\sigma'})$ has the homotopy extension property.
\end{theorem}

\subsection{The lemma} The proof of the sparse multicover nerve theorem hinges
on the lemma below,~\cref{lem:new_cover}, which abstracts a key step in the
second part of the proof of the multicover nerve theorem given in~\cite[Theorem
4.11]{blumbergStability2ParameterPersistent2024}, which is in turn inspired
by~\cite[Lemma 8]{sheehyMulticoverNerveGeometric2012}.

The lemma relies on Quillen's theorem
A~\cite{quillenHigherAlgebraicKtheory}\cite[Proposition
1.6]{quillenHomotopyPropertiesPoset1978}, which we now review. Recall that the
\deff{order complex} $\order(P)$ of a poset $P$ is the simplicial complex given
by the chains of $P$. Any order preserving map $f\colon P\to Q$ induces a
simplicial map $\order(f)\colon \order(P) \to \order(Q)$ by applying $f$ to the
vertices. A simplicial complex $S$ is homeomorphic to the order complex of its
face poset $\order(S)$ (its barycentric subdivision); a fact we tacitly use.

For a $q$ in a poset $Q$, we write
$(\downset{q}{Q}) \coloneqq\Set{q'\in Q\given q'\leq q}$ for its down set.
\begin{theorem}[Quillen's theorem A]\label{thm:quillen}
  Let $f\colon P\to Q$ be an order preserving map between finite posets. If for every
  $q\in Q$ the complex $\order(f^{-1}(\downset{q}{Q}))$ is contractible, then the induced
  map $\order(f)\colon\order(P)\to\order(Q)$ is a homotopy equivalence.
\end{theorem}

A subposet $P\subset N(A)$ of the (face poset of the) nerve of a cover
$A = (A_{i})_{i\in I}$ induces a cover $A^{P} = (A_{\sigma})_{\sigma\in P}$.
The nerve of this cover $N(A^{P})$ is related to $P$:

\begin{lemma}[Subposet nerve lemma]\label{lem:new_cover}
  Let $A = (A_{i})_{i\in I}$ be a finite cover. Let $P\subset N(A)$ be a
  subposet whose induced nerve $N(A^{P})$ satisfies that for every
  $\Set{\sigma_{1}, \dots, \sigma_{m}}\in N(A^{P})$ one has
  $\sigma_{1}\cup\dots\cup\sigma_{m}\in P$.

  There is an order-preserving
  map $F\colon N(A^{P})\to P$ that induces a homotopy equivalence
  $\order(F)\colon \order(N(A^{P})) \to \order(P)$.
\end{lemma}
\begin{proof}
  Define the order-preserving map $F\colon N(A^{P})\to P$ by
  $F(\Set{\sigma_{1}, \dots, \sigma_{m}}) = \sigma_{1}\cup\dots\cup\sigma_{m}$.
  We claim that for every $\sigma\in P$ the complex
  $\order(F^{-1}(\downset{\sigma}{P}))$ is contractible, and so the result
  follows from Quillen's theorem A,~\cref{thm:quillen}. The claim follows from
  the observation that $F^{-1}(\downset{\sigma}{P})$ has a maximum, using the
  standard fact that if a poset has a maximum then its order complex is
  contractible. Indeed, the maximum is
  $(\downset{\sigma}{P}) \in F^{-1}(\downset{\sigma}{P})$.

  To check this, first note that $(\downset{\sigma}{P})$ is in
  $F^{-1}(\downset{\sigma}{P})$ because
  $\varnothing\neq A_{\sigma}\subset A_{\sigma'}$ for every
  $\sigma' \in(\downset{\sigma}{P})$ and thus
  $A_{\sigma} \subset \bigcap_{\sigma'\in (\downset{\sigma}{P})}A_{\sigma'}$, so
  $(\downset{\sigma}{P}) \in N(A^{P})$. In addition, suppose that
  $\Set{\sigma_{1}, \dots, \sigma_{m}} \in F^{-1}(\downset{\sigma}{P})$, which
  means that $\sigma_{1}\cup\dots\cup\sigma_{m} \subset \sigma'$ for some
  $\sigma'\in(\downset{\sigma}{P})$ and, in particular,
  $\sigma_{j}\subset\sigma$ for every $j\in\Set{1,\dots, m}$. Therefore
  $\Set{\sigma_{1}, \dots, \sigma_{m}} \subset (\downset{\sigma}{P})$ and we
  conclude that $\downset{\sigma}{P}$ is the maximum, as desired.
\end{proof}

Combining the above with the nerve theorem we obtain:

\begin{lemma}\label{lem:func_new_cover}
  Let $A = (A_{i})_{i\in I}$ and $A' = (A'_{j})_{j\in J}$ be two finite covers
  with $I\subset J$ and $A_{i}\subset A'_{i}$ for every $i\in I$. If
  $P\subset P'$ are subposets $P\subset N(A)$ and $P'\subset N(A')$ that satisfy
  the conditions of~\cref{lem:new_cover} and if $A$ and $A'$ satisfy the
  conditions of the nerve theorem (either~\cref{thm:nerve} or~\cref{thm:general_nerve}), then we have a commutative diagram
\begin{equation*}
  \begin{tikzcd}
    \bigcup_{\sigma\in P} A_{\sigma} \ar[hook]{d} & \lar B(A^{P}) \rar \ar[hook]{d} & \order(P) \ar[hook]{d}\\
    \bigcup_{\sigma\in P'} A'_{\sigma} & \lar B(A'^{P'}) \rar & \order(P'),
  \end{tikzcd}
\end{equation*}
where the horizontal arrows are homotopy equivalences.
\end{lemma}
\begin{proof}
  If the cover $A$ satisfies the conditions of the nerve theorem, so does the
  cover $A^{P}$, because
  $\bigcap_{\sigma_{i}\in\Set{\sigma_{1}, \dots, \sigma_{m}}} A_{\sigma_{i}} = A_{\sigma_{1}\cup\cdots\cup\sigma_{m}}$.
  So the nerve theorem and~\cref{lem:new_cover} give a
  diagram
\begin{equation*}
  \begin{tikzcd}
    \bigcup_{\sigma\in P} A_{\sigma} \ar[hook]{d} & \lar B(A^{P}) \rar \ar[hook]{d} & \order(N(A^{P})) \rar{F} \ar[hook]{d} & \order(P) \ar[hook]{d}\\
    \bigcup_{\sigma\in P'} A'_{\sigma} & \lar B(A'^{P'}) \rar & \order(N(A'^{P'})) \rar{F'} & \order(P'),
  \end{tikzcd}
\end{equation*}
where the horizontal arrows are homotopy equivalences and the right square can
be seen to be commutative by the definition of $F$ in the proof
of~\cref{lem:new_cover}. Thus, the whole diagram is commutative because the two
leftmost squares commute, by~\cref{thm:nerve}.
\end{proof}

\subsection{Proving the theorem}
The rest of the section is dedicated to proving the sparse multicover nerve
theorem.
\firstsparsenerve*

The strategy is to apply~\cref{lem:new_cover} to a cover of~\emph{cones}, as in
the work of Cavanna, Jahanseir and
Sheehy~\cite{cavannaGeometricPerspectiveSparse2015}.
For a fixed $r\geq 0$, let $C = (C_{x})_{x\in X}$ be the cover given by the
cones of the sparse balls, that is,
\begin{equation*}
  C_{x} \coloneqq \bigcup_{0\leq s \leq r} \SB(x, s) \times\Set{s},
\end{equation*}
a subset of $\R^{d}\times \R$. We now suppose that the radius function
$\rho_{x}$ for the sparse balls is such that the cone $C_{x}$ is
closed and convex. This is the case if the chosen radius function $\rho_{x}$ is concave,
by~\cite[Proposition 4]{cavannaGeometricPerspectiveSparse2015}; one can choose, for example,
\begin{equation*}
  \rho_{x}(r) = \begin{cases}
    r, & r < \slow(x),\\
    U_{x}(r) = \frac{1}{3(1+\epsilon)} r +  \frac{2 + 3\epsilon}{3(1+\epsilon)}\slow(x), & r \geq \slow(x).
  \end{cases}
\end{equation*}
We leave the general case for later.

The nerve of this cover, $N(C)$, is equal to $\SCech(r,1)$:
\begin{equation*}
  N(C) = \SCech(r,1) = \bigcup_{s\leq r} \Set*{\sigma\subset X \given \bigcap_{x\in\sigma} \SB(x, s) \neq \varnothing}.
\end{equation*}
Now fix a $k\in\N$ and let $P \coloneqq \SCech(r,k)$, noting that
$P \subset N(C) = \SCech(r,1)$, and consider $N(C^{P})$. In order to
apply~\cref{lem:func_new_cover} we have the following:

\begin{lemma}
  If $\Set{\sigma_{1}, \dots, \sigma_{m}}\in N(C^{P})$ then
  $\sigma_{1}\cup\dots\cup\sigma_{m} \in P$.
\end{lemma}
\begin{proof}
  Consider $\Set{\sigma_{1}, \dots, \sigma_{m}}\in N(C^{P})$ and let
  $\tau\coloneqq\sigma_{1}\cup\dots\cup\sigma_{m}$. To certify that $\tau\in P$
  we need to find a $\gamma\leq r$ with $\bigcap_{x\in\tau}\SB(x, \gamma) \neq \varnothing$ and
  $\sum_{x\in\tau} c_{\gamma}(x) \geq k$.

  Let $\gamma$ be the maximum $\gamma\leq r$ such that
  $\bigcap_{x\in\tau}\SB(x, \gamma)\neq \varnothing$. This maximum is
  well-defined precisely because
  $\Set{\sigma_{1}, \dots, \sigma_{m}}\in N(C^{P})$, which implies that
  $\bigcap_{x\in\tau}\SB(x, s)\neq\varnothing$ for at least one $s\leq r$.

  If $\gamma = r$ then for any
  $\sigma_{i}\in\Set{\sigma_{1}, \dots, \sigma_{m}}$ we have that
  $\sum_{x\in\sigma_{i}} c_{\gamma}(x) \geq k$ because
  $\sigma_{i}\in P$. It follows that
  $\sum_{x\in\tau} c_{\gamma}(x) \geq \sum_{x\in\sigma_{i}}c_{\gamma}(x) \geq k$,
  as desired.

  If $\gamma < r$ then the sparse ball of a point in $\tau$ disappears precisely
  at $\gamma$, meaning that there is a $y\in\tau$ such that $(1+3\epsilon)\slow(y) = \gamma$. Let
  $\sigma_{i}\in\Set{\sigma_{1}, \dots, \sigma_{m}}$ be such that
  $y\in\sigma_{i}$. Then, from $\sigma_{i}\in P = \SCech(r,k)$, we know that
  $\sum_{x\in\sigma_{i}} c_{\gamma}(x) \geq k$ and therefore
  $\sum_{x\in\tau} c_{\gamma}(x) \geq \sum_{x\in\sigma_{i}}c_{\gamma}(x) \geq k$,
  as desired.
\end{proof}

Note that $\SCov(r,k)\times\Set{r} \subset \bigcup_{\sigma\in P} C_{\sigma}$
and, in fact, such an inclusion is a homotopy equivalence, which follows from
$\bigcup_{\sigma\in P} C_{\sigma}$ having $\SCov(r,k)\times\Set{r}$ as a
deformation retract, by the homotopy that slides each
$(p, s)\in\bigcup_{\sigma\in P} C_{\sigma}$ to
$(p, r) \in \SCov(r,k)\times\Set{r}$ in a straight line. Note that the slice of
a cone, meaning $C_{\sigma} \cap (\R^{d}\times\Set{s})$ for a $s\leq r$, might
turn empty---$C_{\sigma} \cap (\R^{d}\times\Set{s}) \neq\varnothing$ and
$C_{\sigma} \cap (\R^{d}\times\Set{s'}) = \varnothing$ for some
$s\leq s'$---because sparse balls disappear; but~\cref{prop:filtration} guarantees that the homotopy is well-defined.

Then we can apply~\cref{lem:func_new_cover} and obtain that the bifiltrations
$\SSub$ and $\SCov$ are weakly equivalent.

\subparagraph{The general case.} We have assumed that the cones $C_{x}$ are
closed and convex. This is required by the conditions of the nerve
theorem,~\cref{thm:nerve}. We now drop this assumption by considering a cover
based on \deff{telescopes}, to which the more general nerve
theorem,~\cref{thm:general_nerve}, applies.

We first order the disappearing times
$D = \Set{(1+3\epsilon) \slow(x) \given x \in X}$ of the points in $X$,
obtaining $d_{1} < \cdots < d_{m + 1}$; note that $d_{m+1} = \infty$ because
there is one, and only one, point with $\slow(x) = \infty$. To handle
$d_{m+1} = \infty$ below, we also define the
sequence $d'_{1} < \cdots < d'_{m+1}$ by $d'_{i} \coloneqq d_{i}$ if $i \leq m$, and
$d'_{m+1} \coloneqq d_{m} + 1$.

We fix $r\geq 0$ and let $x\in X$ be a point. Let $d_{k} \in D$ be such that
$d_{k} = (1+3\epsilon)\slow(x)$. If $r > d_{k}$, the \deff{telescope} $T_{x}$
for $x\in X$ is
\begin{equation*}
  T_{x} \coloneqq \left(\bigcup_{j< k} \SB(x, d'_{j})\times [d'_{j}, d'_{j+1}]\right)\cup\SB(x, d'_{k})\times\Set{d'_{k}},
\end{equation*}
otherwise, if $r \leq d_{k}$, then let $d_{i}$ be the minimum disappearing time
$d_{i}\in D$ such that $r\leq d_{i}$ and define
\begin{equation*}
  T_{x} \coloneqq \left(\bigcup_{j < i}\SB(x, d'_{j})\times[d'_{j}, d'_{j+1}]\right) \cup \SB(x, r)\times \Set{d'_{i}},
\end{equation*}
see~\cref{fig:telescopes}. Note that, now writing $T_{x}(r)$ for the telescope
of a fixed $r\geq 0$, if $r\leq r'$ then we have an inclusion
$T_{x}(r)\subset T_{x}(r')$.
\begin{figure}[h]
  \centering
  \includegraphics{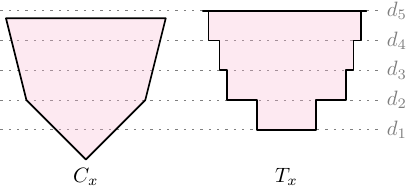}
  \caption{Comparison of a cone $C_{x}\subset\R\times\R$ (in this case convex, but it might
    not be), and a telescope $T_{x}$.}\label{fig:telescopes}
\end{figure}

Each $T_{x}$ is a closed subset of $\R^{d}\times\R$, being the finite union of closed
sets, but not convex. Still, the cover $T = (T_{x})_{x\in X}$ has the
same nerve as the cones $C = (C_{x})_{x\in X}$:
\begin{lemma}
  $N(T) = N(C)$.
\end{lemma}
\begin{proof}
  If $\sigma\in N(T)$ then $\bigcap_{x\in\sigma} \SB(x, s') \neq \varnothing$ for a $s' \leq r$ and thus $\sigma\in N(C)$.

  Conversely, suppose that $\sigma\in N(C)$, meaning that
  $\bigcap_{x\in\sigma} C_{x}\neq \varnothing$ and thus there exists a
  $s\leq r$ such that $\bigcap_{x\in\sigma} \SB(x, s)\neq\varnothing$. Let
  $d_{i}$ be the minimum disappearing time $d_{i}\in D$ such that $s\leq d_{i}$.
  Then, by construction, there is no point $y\in X$ with disappearing time
  $(1+3\epsilon)\slow(y)$ in $[s, d_{i})$; this implies that if $\SB(x,s)\neq\varnothing$ then $\SB(x, d'_{i})\neq\varnothing$. Thus,
  $\SB(x, s)\times\Set{d'_{i}} \subset T_{x}$ and then
  $\bigcap_{x\in\sigma} T_{x}\neq\varnothing$, as desired.
\end{proof}

In fact, $N(T^{P}) = N(C^{P})$ too, for a $P$ defined as before. In addition,
the cover $T = (T_{x})_{x\in X}$ satisfies the conditions of the more general
nerve theorem,~\cref{thm:general_nerve}. Namely, for each $\sigma\subset X$ the
subset $T_{\sigma}$ is either empty or contractible, and for each
$\sigma\subset\sigma'\subset X$ the pair $(T_{\sigma}, T_{\sigma'})$ has the
homotopy extension property, noting that any two closed and convex subsets
$K'\subset K$ of $\R^{d}$ have the homotopy extension property~\cite[Proposition
5.19]{bauerUnifiedViewFunctorial2023}.

The union $\bigcup_{\sigma\in P} T_{\sigma}$ contains $\SCov(r,k)$ as a
deformation retract, for the same reasons as for the union of the cones. Thus,
using the same argument as before, we have that $\SSub$ and $\SCov$ are
weakly equivalent by~\cref{lem:func_new_cover}.

\section{Sparse subdivision Rips}\label{sec:rips}

Although we have been focusing on $\R^{d}$, the sparsification scheme works for
general metric spaces and the Rips complex, as in the single-parameter
case~\cite{cavannaGeometricPerspectiveSparse2015}. Recall that, fixing a (finite)
metric space $(X, \partial)$, the \deff{Rips complex} $\Rips(r)$ of $X$ for a radius
$r\geq 0$ is given by
\begin{equation*}
  \Rips(r) \coloneqq \Set{\sigma\subset X \given \partial(x, x') \leq 2r \text{ for all $x,x'\in \sigma$}}.
\end{equation*}
When $X$ is a subset of $\R^{d}$ under the $l^{\infty}$ norm, then the Rips
complex and the \v{C}ech complex of $X$ are equal.

As explained in the work of Cavanna, Jahanseir and
Sheehy~\cite{cavannaGeometricPerspectiveSparse2015}, their construction for the
\v{C}ech complex and union of balls in $\R^{d}$ is independent of the chosen
norm; therefore, by isometrically embedding the finite metric space $X$ in
$\R^{\abs{X}}$ under the $l^{\infty}$ norm (cf.~\cite[Example
3.5.3]{buragoCourseMetricGeometry2001}) and observing again that in that setting
the Rips complex and \v{C}ech complex are equal, the construction also applies
to the Rips complex. Moreover, assuming further that the doubling dimension of
$X$ is constant (see below), they show that their construction is of linear size
and can be computed in $O(\abs{X}\log\Delta)$ time, where $\Delta$ is the
spread---also in the Rips setting.

For the same reasons, the sparse bifiltrations defined here also extend to the
Rips setting. Recall that we have been working with an arbitrary norm from the
start. Embedding the finite metric space $X$ in $(\R^{\abs{X}}, l^{\infty})$
yields a construction for the \deff{subdivision Rips
bifiltration}~\cite{sheehyMulticoverNerveGeometric2012} (equal in this setting to
the subdivision \v{C}ech bifiltration), which only depends on the pairwise
distances of $X$.

The size bounds of~\cref{sec:size} also apply, assuming constant
doubling dimension. Recall that the \deff{doubling dimension} of a metric space
$X = (X, \partial)$ is the minimum value $d$ such that every ball
$B(x, r) \coloneqq \Set{y\in X \given \partial(x, y) \leq r}$ in $X$ can be
covered by $2^{d}$ balls of half the radius. The following packing property,
see~\cite[Lemma 2]{smidWeakGapProperty2009} for a proof, allows to extend the
packing arguments of~\cref{sec:size} to metric spaces of constant doubling
dimension.
\begin{lemma}
  Suppose that $X$ has doubling dimension $d$. Any subset $S$ that is an
  $r$-packing and contained in a ball of radius $R$ has
  $\abs{S}\leq\left(\frac{4R}{r}\right)^{d}$.
\end{lemma}
As for computation, the bounds for the computation of the greedy
permutation~\cite{clarksonFastAlgorithmsAll1983,har-peledFastConstructionNets2006a}
and friends~\cite{cavannaGeometricPerspectiveSparse2015} also apply to finite metric spaces. Also, in the Rips case the
intersection of balls is only computed pairwise, so for simplicity one can take
the radius function $\rho_{x}$ for each $x\in X$ to be
\begin{equation*}
  \rho_{x}(r) = \begin{cases}
    r, & r < \slow(x),\\
    U_{x}(r) = \frac{1}{3(1+\epsilon)} r +  \frac{3\epsilon + 2}{3(1+\epsilon)}\slow(x), & r \geq \slow(x).
  \end{cases}
\end{equation*}

\appendix

\section{LP-type problems and enclosing balls}\label{sec:lp_type}
In this appendix we explain how to solve the problem~\eqref{eq:problem_min}
of~\cref{sec:intersection_balls} as an LP-type
problem~\cite{matousekSubexponentialBoundLinear1996} in expected time
$O(d^{3}2^{2d}m)$, following the strategy of the smallest enclosing ball of
balls~\cite{fischerSmallestEnclosingBall2004}. In fact, as already mentioned,
the problem~\eqref{eq:problem_min} can be written as the \textit{smallest
  enclosing ball for a point set with strictly convex sets} of
Z{\"u}rcher~\cite{zurcherSmallestEnclosingBall2007}, which itself is an LP-type
problem and an extension of Fischer and G{\"a}rtner's
solution~\cite{fischerSmallestEnclosingBall2004}. Here we briefly recall LP-type
problems and prove the missing lemmas, all similar to the lemmas of
Fischer and G{\"a}rtner~\cite{fischerSmallestEnclosingBall2004}, needed to solve~\eqref{eq:problem_min}.

A \deff{LP-type problem} is a
pair $(H, w)$ given by a finite set of constraints $H$ and a function
$w\colon 2^{H}\to\R\cup\Set{-\infty}$, understood as mapping every subset of $H$
to its optimal solution, that satisfies the following for all subsets
$F\subset G\subset H$
\begin{enumerate}
  \item \textit{monotonicity:} $w(F)\leq w(G)$, and
  \item \textit{locality:} if $w(F) < w(G)$ there is an $h\in G$ such that
        $w(F) < w(F \cup\Set{h})$.
\end{enumerate}

A \deff{basis $F$ of $G$} is an inclusion-minimal nonempty subset $F\subset G$
such that $w(F) = w(G)$. A subset $F\subset H$ is a \deff{basis} if it is a
basis of itself.

The problem is solved by finding a basis of $H$, the set of all constraints. The MSW
algorithm~\cite{matousekSubexponentialBoundLinear1996}, viewed here as a black
box, does so given the following primitives:
\begin{enumerate}
  \item \textit{violation test:} given a constraint $h$ and a basis $F$, test
        whether $w(F) < w(F\cup\Set{h})$, and
  \item \textit{basis computation:} given a constraint $h$ and a basis $F$,
        compute a basis of $F\cup\Set{h}$.
\end{enumerate}
The expected number of violation tests and basis computations performed by the
MSW algorithm is $O(2^{\delta}\abs{H})$~\cite{matousekSubexponentialBoundLinear1996},
where $\delta$ is the \deff{combinatorial dimension} of the problem: the size of
the largest basis.

To view~\eqref{eq:problem_min} as an LP-type problem, for a point $z\in \R^{d}$
and a constraint given by the triple $(p_{i}, \alpha_{i}, \beta_{i})$ where, as above,
$p_{i}\in\R^{d}$, $\alpha_{i}$ is a positive real number, and $\beta_{i}$ is a
non-negative real number, we define the \deff{reach} of $z$ with respect to the
constraint by
\begin{equation*}
  \reach(z, (p_{i}, \alpha_{i}, \beta_{i})) \coloneqq \frac{1}{\alpha_{i}}(\norm{p_{i} - z}^{2} - \beta_{i}).
\end{equation*}
For a subset $F\subset H$ of the constraints $H$ we write
\begin{equation*}
  \reach_{F}(z) \coloneqq \max_{(p_{i}, \alpha_{i}, \beta_{i})\in F} \reach(z, (p_{i}, \alpha_{i}, \beta_{i})).
\end{equation*}
We define
$w\colon 2^{H}\to\R\cup\Set{-\infty}$ by taking a subset of the constraints $F\subset H$ to
\begin{equation}\label{eq:def_w}
  w(F) \coloneqq \min_{z\in\R^{d}}\Gamma_{F}(z),
\end{equation}
with the convention that $w(\varnothing) = -\infty$. The solution
of~\eqref{eq:problem_min} is given by $w(H)$ and the associated minimizing
$z\in\R^{d}$.

If we would be computing the minimum enclosing ball (that is, with
$\alpha_{i} = 1$ and $\beta_{i} = 0$) then the $z\in\R^{d}$ that minimizes
$w(H)$ is the center of such a ball. In general, if $z\in\R^{d}$ is one of the
points that minimizes $w(F)$ in \cref{eq:def_w} above, we say that $z$ is a
\textit{center} of $F$. Similarly, we say
that a point $z\in \R^{d}$ is \deff{tangent} to $F\subset H$ if
$\Gamma(z, h) = \Gamma_{F}(z)$ for all constraints $h\in F$.

For each constraint $h\coloneqq (p_{i}, \alpha_{i}, \beta_{i})$, the reach
$\reach(z, h)$ is \deff{strictly convex} as a function of $z$, meaning
$\reach(t z_{1} + (1-t) z_{2}, h) < t \reach(z_{1}, h) + (1 - t) \reach(z_{2}, h)$
when $t\in (0, 1)$ and $z_{1}\neq z_{2}$. This follows from $\norm{x}^{2}$ being
strictly convex. In turn, this allows to write~\eqref{eq:problem_min} as the,
more general, \textit{smallest enclosing ball for a point set with strictly
  convex sets} of Zürcher~\cite{zurcherSmallestEnclosingBall2007}. The following
two lemmas, and that $w$ is in fact a LP-type problem, then follow
from~\cite{zurcherSmallestEnclosingBall2007}; for completeness, we give, shorter
because our case is more specific, proofs.

\begin{lemma}\label{lem:lp_center_unique}
  The center of a subset of the constraints $F\subset H$ exists and is unique.
\end{lemma}
\begin{proof}
  Existence follows from a standard compactness argument.
  Uniqueness is a standard strict convexity argument:
   $\Gamma_{F}(z)$ is strictly convex, because it is the maximum of
  strictly convex functions. Now, if there are two different centers, $z_{1}$ and $z_{2}$,
  we can take $z'\coloneqq \frac{1}{2}z_{1} + \frac{1}{2}z_{2}$ and conclude that
  $\Gamma_{F}(z') < \Gamma_{F}(z_{1})$, by the strict convexity of
  $\Gamma_{F}(z)$. This contradicts the minimality of $z_{1}$ and $z_{2}$.
\end{proof}

\begin{lemma}\label{lem:lp_equality}
  The center $z_{F}$ of a basis $F$ is tangent to $F$.
\end{lemma}
\begin{proof}
  Let $z_{F}$ be the center of $F$ and let $G\subset F$ be the subset of
  constraints $g\in F$ with $\Gamma(z_{F}, g) < \Gamma_{F}(z_{F})$. We argue by
  contradiction that $G$ is empty. Let $y$ be the center of $F\setminus G$. If
  $G$ is non-empty, we have that $y\neq z_{F}$, because otherwise
  $w(F\setminus G) = w(F)$, contradicting that $F$ is a basis.

  Because $\Gamma_{F}(z)$ is a
  strictly convex function, for every $t\in(0,1)$ we have
  \begin{equation*}
    \Gamma_{F\setminus G}(t z_{F} + (1-t)y) < t\Gamma_{F\setminus G}(z_{F}) + (1-t)\Gamma_{F\setminus G}(y) \leq w(F).
  \end{equation*}
  Moreover, for all $g\in G$,
  $\Gamma(t z_{F} + (1-t)y, g) < t\Gamma(z_{F}, g) + (1-t)\Gamma(y, g)$, and
  thus there exists $t'\in (0,1)$ such that
  $\Gamma(t' z_{F} + (1-t')y, g) < \Gamma_{F}(z_{F})$, because $\Gamma(z_{F}, g) < \Gamma_{F}(z_{F})$.
  All in all, it follows that $\Gamma_{F}(t' z_{F} + (1-t')y) < \Gamma_{F}(z_{F})$,
  contradicting that $z_{F}$ is the center.
\end{proof}

We now compute the center of a basis, with the aim of implementing the
violation test and the basis computation primitives of the MSW algorithm. For that, we need the following.

\begin{lemma}\label{lem:point_unique}
  Let $F$ be a basis. The points in $F$ are
  unique, meaning there are no two different
  $(p_{i}, \alpha_{i}, \beta_{i}), (p_{j},\alpha_{j},\beta_{j})\in F$ with
  $p_{i} = p_{j}$.
\end{lemma}
\begin{proof}
  Suppose, on the contrary, that there are two different
  $h \coloneqq (p_{i}, \alpha_{i}, \beta_{i})$ and
  $h'\coloneqq (p_{j}, \alpha_{j}, \beta_{j})$ in $F$ with $p = p_{i} = p_{j}$.
  Let $z_{F}$ and $y$ be the centers of $F$ and $F\setminus \Set{h}$,
  respectively. We have that $\Gamma_{F\setminus\Set{h}}(y) < \Gamma_{F}(z_{F})$
  and $y\neq z_{F}$, because otherwise $F$ cannot be a basis. But then it
  follows that $\norm{p - y}^{2} < \norm{p - z_{F}}^{2}$, which implies that
  $z_{F}$ is not tangent to $F$ and thus contradicts that $z_{F}$ is the center
  of $F$, by~\cref{lem:lp_equality}.
\end{proof}

The proof of the following lemma follows by an adaptation of the analogous
result for the smallest enclosing ball of balls~\cite[Lemma
5]{fischerSmallestEnclosingBall2004}, as shown in~\cite[Lemma
5.3]{zurcherSmallestEnclosingBall2007}.

\begin{lemma}\label{lem:basis_characterization}
  The center $z_{F}$ of a basis $F\subset H$ is in the convex hull of the points
  $\Set{p \given (p, \alpha, \beta)\in F}$ of $F$, and these are affinely
  independent.

  If a point $z_{F}$ is tangent to $F\subset H$ and is in the convex hull of
  the points $\Set{p\given (p, \alpha, \beta)\in F}$ of $F$, then $z_{F}$ is the
  center of $F$.
\end{lemma}

The following is similar to~\cite[Lemma 8]{fischerSmallestEnclosingBall2004}.

\begin{lemma}\label{lem:center_computation}
  The center $z_{F}$ of a basis $F\subset H$ can be computed in $O(d^{3})$-time.
\end{lemma}
\begin{proof}
  Setting $s\coloneqq w(F)$, for each of the $k$ constraints
  $h_{1}, \dots, h_{k}$ in $F$, with $h_{i} = (p_{i}, \alpha_{i}, \beta_{i})$,
  we have $\norm{p_{i} - z_{F}}^{2} = \alpha_{i} s + \beta_{i}$,
  by~\cref{lem:lp_equality}. Moreover, by~\cref{lem:point_unique}, all the
  $p_{1}, \dots, p_{k}$ are unique.

  Now, under the change of variables $c\coloneqq z_{F} - p_{1}$ and
  $q_{i} \coloneqq p_{i} - p_{1}$, we write
  \begin{align}\label{eq:quadratic}
    c^{T}c & = \alpha_{1} s + \beta_{1},\\
    (q_{i} - c)^{T} (q_{i} - c) &= \alpha_{i} s + \beta_{i}, \text{ for $i = 2,\dots, k$}.\nonumber
  \end{align}
  Subtracting the second from the first, we obtain
  \begin{equation*}
    2q_{i}^{T}c = q_{i}^{T}q_{i} + (\alpha_{1} - \alpha_{i})s + (\beta_{1} - \beta_{i}), \text{ for $i = 2,\dots, k$}.
  \end{equation*}
  Using the fact that $z_{F}$ is in the convex hull of the $p_{1}, \dots, p_{k}$,
  and therefore can be written as $z_{F} = \sum_{j=1}^{k}\lambda_{j}p_{j}$ with
  $\sum_{j=1}^{k}\lambda_{j} = 1$, we can write
  \begin{equation*}
    c = \sum_{j=1}^{k}\lambda_{j} p_{j} - \sum_{j=1}^{k}\lambda_{j}p_{1} = \sum_{j=2}^{k}\lambda_{j}q_{j}.
  \end{equation*}
  Letting $Q = (q_{2}, \dots, q_{k})$ and
  $\lambda = (\lambda_{2}, \dots, \lambda_{k})$, we obtain the following linear
  equations
  \begin{equation*}
    2 q_{i}^{T} Q\lambda = q_{i}^{T}q_{i} + (\alpha_{1} - \alpha_{i})s + (\beta_{1} - \beta_{i}), \text{ for $i = 2, \dots, k$},
  \end{equation*}
  that assemble into the linear system $2Q^{T}Q\lambda = E + Ds$. Since the
  points $p_{1}, \dots, p_{k}$ are affinely independent,
  by~\cref{lem:basis_characterization}, the matrix $Q^{T}Q$ is nonsingular, and
  we can obtain a parametrization of $\lambda$ in terms of $s$.

  Finally, we substitute back in the quadratic equations of~\cref{eq:quadratic}.
  Only one of the solutions obtained is valid.
\end{proof}

To do the violation test of a constraint $h$ and a basis $F$, we compute the
center $z_{F}$ of $F$ and check whether $w(F) < \reach(z_{F}, h)$. To do the basis
computation, given a constraint $h$ and a basis $G$, we do as in the smallest
enclosing ball of balls~\cite{fischerSmallestEnclosingBall2004}: we iterate over
all subsets $F \subset G\cup\Set{h}$ with $h\in F$, in increasing order of size.
We check if each $F$ is a basis by doing the computations
of~\cref{lem:center_computation}, which, if applicable and successful, yield a
point $z_{F}$ that is tangent to $F$ and in the convex hull of the points of
$F$; by~\cref{lem:basis_characterization} such an $F$ is a basis. If, in
addition, we have that $\Gamma(z_{F}, g) \leq w(F)$ for every
$g\in G\cup\Set{h}$ then $F$ is a basis of $G\cup\Set{h}$, by the chosen enumeration order.

All in all, because the combinatorial dimension (the size of the largest basis)
is $d+1$, by~\cref{lem:basis_characterization}, the MSW algorithm solves the
problem in time $O(d^{3}2^{2d}m)$, as claimed.

\bibliographystyle{plainurl_fulljournal}
\bibliography{refs}

\begin{thebibliography}{10}

\bibitem{anaiDTMBasedFiltrations2019}
Hirokazu Anai, Fr{\'e}d{\'e}ric Chazal, Marc Glisse, Yuichi Ike, Hiroya
  Inakoshi, Rapha{\"e}l Tinarrage, and Yuhei Umeda.
\newblock {{DTM-Based Filtrations}}.
\newblock In Gill Barequet and Yusu Wang, editors, {\em 35th {{International
  Symposium}} on {{Computational Geometry}}, {{SoCG}} 2019}, volume 129 of {\em
  {{LIPIcs}}}, pages 58:1--58:15, 2019.
\newblock \href {https://doi.org/10.4230/LIPIcs.SoCG.2019.58}
  {\path{doi:10.4230/LIPIcs.SoCG.2019.58}}.

\bibitem{anaiDTMbasedFiltrations2020}
Hirokazu Anai, Fr{\'e}d{\'e}ric Chazal, Marc Glisse, Yuichi Ike, Hiroya
  Inakoshi, Rapha{\"e}l Tinarrage, and Yuhei Umeda.
\newblock {{DTM-Based Filtrations}}.
\newblock In Nils~A. Baas, Gunnar~E. Carlsson, Gereon Quick, Markus Szymik, and
  Marius Thaule, editors, {\em Topological {{Data Analysis}}}, pages 33--66.
  Springer, 2020.
\newblock \href {https://doi.org/10.1007/978-3-030-43408-3_2}
  {\path{doi:10.1007/978-3-030-43408-3_2}}.

\bibitem{attaliVietorisripsComplexesAlso2011}
Dominique Attali, Andr{\'e} Lieutier, and David Salinas.
\newblock Vietoris-{{Rips}} complexes also provide topologically correct
  reconstructions of sampled shapes.
\newblock In {\em Proceedings of the Twenty-Seventh Annual Symposium on
  {{Computational}} Geometry, {{SoCG}} 2011}, pages 491--500, June 2011.
\newblock \href {https://doi.org/10.1145/1998196.1998276}
  {\path{doi:10.1145/1998196.1998276}}.

\bibitem{attaliVietorisRipsComplexes2013}
Dominique Attali, Andr{\'e} Lieutier, and David Salinas.
\newblock Vietoris--{{Rips}} complexes also provide topologically correct
  reconstructions of sampled shapes.
\newblock {\em Computational Geometry}, 46(4):448--465, May 2013.
\newblock \href {https://doi.org/10.1016/j.comgeo.2012.02.009}
  {\path{doi:10.1016/j.comgeo.2012.02.009}}.

\bibitem{bauerUnifiedViewFunctorial2023}
Ulrich Bauer, Michael Kerber, Fabian Roll, and Alexander Rolle.
\newblock A unified view on the functorial nerve theorem and its variations.
\newblock {\em Expositiones Mathematicae}, 41(4):125503, December 2023.
\newblock \href {https://doi.org/10.1016/j.exmath.2023.04.005}
  {\path{doi:10.1016/j.exmath.2023.04.005}}.

\bibitem{blumbergRobustStatisticsHypothesis2014}
Andrew~J. Blumberg, Itamar Gal, Michael~A. Mandell, and Matthew Pancia.
\newblock Robust {{Statistics}}, {{Hypothesis Testing}}, and {{Confidence
  Intervals}} for {{Persistent Homology}} on {{Metric Measure Spaces}}.
\newblock {\em Foundations of Computational Mathematics}, 14:745--789, August
  2014.
\newblock \href {https://doi.org/10.1007/s10208-014-9201-4}
  {\path{doi:10.1007/s10208-014-9201-4}}.

\bibitem{blumbergUniversalityHomotopyInterleaving2023}
Andrew~J. Blumberg and Michael Lesnick.
\newblock Universality of the homotopy interleaving distance.
\newblock {\em Transactions of the American Mathematical Society},
  376:8269--8307, 2023.
\newblock \href {https://doi.org/10.1090/tran/8738}
  {\path{doi:10.1090/tran/8738}}.

\bibitem{blumbergStability2ParameterPersistent2024}
Andrew~J. Blumberg and Michael Lesnick.
\newblock Stability of 2-{{Parameter Persistent Homology}}.
\newblock {\em Foundations of Computational Mathematics}, 24:385--427, 2024.
\newblock \href {https://doi.org/10.1007/s10208-022-09576-6}
  {\path{doi:10.1007/s10208-022-09576-6}}.

\bibitem{bobrowskiTopologicalConsistencyKernel2017}
Omer Bobrowski, Sayan Mukherjee, and Jonathan~E. Taylor.
\newblock Topological consistency via kernel estimation.
\newblock {\em Bernoulli}, 23(1):288--328, February 2017.
\newblock \href {https://doi.org/10.3150/15-BEJ744}
  {\path{doi:10.3150/15-BEJ744}}.

\bibitem{buchetSparseHigherOrder2023}
Micka{\"e}l Buchet, Bianca B.~Dornelas, and Michael Kerber.
\newblock Sparse {{Higher Order {\v C}ech Filtrations}}.
\newblock In Erin~W. Chambers and Joachim Gudmundsson, editors, {\em 39th
  {{International Symposium}} on {{Computational Geometry}}, {{SoCG}} 2023},
  volume 258 of {\em {{LIPIcs}}}, pages 20:1--20:17, 2023.
\newblock \href {https://doi.org/10.4230/LIPIcs.SoCG.2023.20}
  {\path{doi:10.4230/LIPIcs.SoCG.2023.20}}.

\bibitem{buchetEfficientRobustPersistent2015}
Micka{\"e}l Buchet, Fr{\'e}d{\'e}ric Chazal, Steve~Y. Oudot, and Donald~R.
  Sheehy.
\newblock Efficient and robust persistent homology for measures.
\newblock In {\em Proceedings of the Twenty-Sixth Annual {{ACM-SIAM}} Symposium
  on {{Discrete}} Algorithms, {{SODA}} 2015}, pages 168--180, January 2015.

\bibitem{buchetEfficientRobustPersistent2016}
Micka{\"e}l Buchet, Fr{\'e}d{\'e}ric Chazal, Steve~Y. Oudot, and Donald~R.
  Sheehy.
\newblock Efficient and robust persistent homology for measures.
\newblock {\em Computational Geometry}, 58:70--96, October 2016.
\newblock \href {https://doi.org/10.1016/j.comgeo.2016.07.001}
  {\path{doi:10.1016/j.comgeo.2016.07.001}}.

\bibitem{buchetSparseHigherOrder2024}
Micka{\"e}l Buchet, Bianca~B. Dornelas, and Michael Kerber.
\newblock Sparse {{Higher Order {\v C}ech Filtrations}}.
\newblock {\em Journal of the ACM}, 71(4):1--23, August 2024.
\newblock \href {https://doi.org/10.1145/3666085} {\path{doi:10.1145/3666085}}.

\bibitem{buragoCourseMetricGeometry2001}
Dmitri Burago, Yuri Burago, and Sergei Ivanov.
\newblock {\em A {{Course}} in {{Metric Geometry}}}, volume~33 of {\em Graduate
  {{Studies}} in {{Mathematics}}}.
\newblock American Mathematical Society, Providence, Rhode Island, June 2001.
\newblock \href {https://doi.org/10.1090/gsm/033} {\path{doi:10.1090/gsm/033}}.

\bibitem{carlssonTheoryMultidimensionalPersistence2009}
Gunnar Carlsson and Afra Zomorodian.
\newblock The {{Theory}} of {{Multidimensional Persistence}}.
\newblock {\em Discrete \& Computational Geometry}, 42(1):71--93, July 2009.
\newblock \href {https://doi.org/10.1007/s00454-009-9176-0}
  {\path{doi:10.1007/s00454-009-9176-0}}.

\bibitem{cavannaWhenWhyTopological2017}
Nicholas~J. Cavanna, Kirk~P. Gardner, and Donald~R. Sheehy.
\newblock When and {{Why}} the {{Topological Coverage Criterion Works}}.
\newblock In {\em Proceedings of the {{Twenty-Eighth Annual ACM-SIAM
  Symposium}} on {{Discrete Algorithms}}, {{SODA}} 2017}, pages 2679--2690,
  January 2017.
\newblock \href {https://doi.org/10.1137/1.9781611974782.177}
  {\path{doi:10.1137/1.9781611974782.177}}.

\bibitem{cavannaGeometricPerspectiveSparse2015}
Nicholas~J. Cavanna, Mahmoodreza Jahanseir, and Donald~R. Sheehy.
\newblock A {{Geometric Perspective}} on {{Sparse Filtrations}}.
\newblock In {\em Proceedings of the 27th {{Canadian Conference}} on
  {{Computational Geometry}}}, pages 116--121, 2015.
\newblock URL: \url{http://arxiv.org/abs/1506.03797}, \href
  {https://arxiv.org/abs/1506.03797} {\path{arXiv:1506.03797}}.

\bibitem{cavannaVisualizingSparseFiltrations2015}
Nicholas~J. Cavanna, Mahmoodreza Jahanseir, and Donald~R. Sheehy.
\newblock Visualizing {{Sparse Filtrations}}.
\newblock In Lars Arge and J{\'a}nos Pach, editors, {\em 31st {{International
  Symposium}} on {{Computational Geometry}}, {{SoCG}} 2015}, volume~34, pages
  23--25. Schloss Dagstuhl -- Leibniz-Zentrum f{\"u}r Informatik, 2015.
\newblock \href {https://doi.org/10.4230/LIPICS.SOCG.2015.23}
  {\path{doi:10.4230/LIPICS.SOCG.2015.23}}.

\bibitem{cawoodWeightedEuclideanOnecenter2024}
Mark~E. Cawood and P.~M. Dearing.
\newblock The weighted {{Euclidean}} one-center problem in {$R^n$}.
\newblock {\em Computational Optimization and Applications}, August 2024.
\newblock \href {https://doi.org/10.1007/s10589-024-00599-z}
  {\path{doi:10.1007/s10589-024-00599-z}}.

\bibitem{chazalGeometricInferenceProbability2011}
Fr{\'e}d{\'e}ric Chazal, David {Cohen-Steiner}, and Quentin M{\'e}rigot.
\newblock Geometric {{Inference}} for {{Probability Measures}}.
\newblock {\em Foundations of Computational Mathematics}, 11(6):733--751,
  December 2011.
\newblock \href {https://doi.org/10.1007/s10208-011-9098-0}
  {\path{doi:10.1007/s10208-011-9098-0}}.

\bibitem{chazalRobustTopologicalInference2017}
Fr{\'e}d{\'e}ric Chazal, Brittany Fasy, Fabrizio Lecci, Bertrand Michel,
  Alessandro Rinaldo, and Larry Wasserman.
\newblock Robust topological inference: Distance to a measure and kernel
  distance.
\newblock {\em J. Mach. Learn. Res.}, 18(1):5845--5884, January 2017.

\bibitem{chazalPersistenceBasedClusteringRiemannian2013}
Fr{\'e}d{\'e}ric Chazal, Leonidas~J. Guibas, Steve~Y. Oudot, and Primoz Skraba.
\newblock Persistence-{{Based Clustering}} in {{Riemannian Manifolds}}.
\newblock {\em Journal of the ACM}, 60(6):1--38, November 2013.
\newblock \href {https://doi.org/10.1145/2535927} {\path{doi:10.1145/2535927}}.

\bibitem{chubetProximitySearchGreedy2023}
Oliver~A. Chubet, Parth Parikh, Donald~R. Sheehy, and Siddharth Sheth.
\newblock Proximity {{Search}} in the {{Greedy Tree}}.
\newblock In {\em 2023 {{Symposium}} on {{Simplicity}} in {{Algorithms}}
  ({{SOSA}})}, January 2023.
\newblock \href {https://doi.org/10.1137/1.9781611977585}
  {\path{doi:10.1137/1.9781611977585}}.

\bibitem{clarksonFastAlgorithmsAll1983}
Kenneth~L. Clarkson.
\newblock Fast algorithms for the all nearest neighbors problem.
\newblock In {\em 24th {{Annual Symposium}} on {{Foundations}} of {{Computer
  Science}}, {{SFCS}} 1983}, pages 226--232, November 1983.
\newblock \href {https://doi.org/10.1109/SFCS.1983.16}
  {\path{doi:10.1109/SFCS.1983.16}}.

\bibitem{clarksonNearestNeighborSearching2003}
Kenneth~L. Clarkson.
\newblock Nearest neighbor searching in metric spaces: {{Experimental}} results
  for sb({{S}}), 2003.
\newblock URL: \url{https://kenclarkson.org/Msb/white_paper.pdf}.

\bibitem{cohen-steinerStabilityPersistenceDiagrams2005}
David {Cohen-Steiner}, Herbert Edelsbrunner, and John Harer.
\newblock Stability of persistence diagrams.
\newblock In {\em Proceedings of the Twenty-First Annual Symposium on
  {{Computational}} Geometry, {{SoCG}} 2005}, {{SoCG}} '05, pages 263--271,
  June 2005.
\newblock \href {https://doi.org/10.1145/1064092.1064133}
  {\path{doi:10.1145/1064092.1064133}}.

\bibitem{cohen-steinerStabilityPersistenceDiagrams2007}
David {Cohen-Steiner}, Herbert Edelsbrunner, and John Harer.
\newblock Stability of {{Persistence Diagrams}}.
\newblock {\em Discrete \& Computational Geometry}, 37(1):103--120, January
  2007.
\newblock \href {https://doi.org/10.1007/s00454-006-1276-5}
  {\path{doi:10.1007/s00454-006-1276-5}}.

\bibitem{corbetComputingMulticoverBifiltration2021conf}
Ren{\'e} Corbet, Michael Kerber, Michael Lesnick, and Georg Osang.
\newblock Computing the multicover bifiltration.
\newblock In Kevin Buchin and {\'E}ric {Colin de Verdi{\`e}re}, editors, {\em
  37th International Symposium on Computational Geometry, {{SoCG}} 2021},
  volume 189 of {\em {{LIPIcs}}}, pages 27:1--27:17, 2021.
\newblock \href {https://doi.org/10.4230/LIPIcs.SoCG.2021.27}
  {\path{doi:10.4230/LIPIcs.SoCG.2021.27}}.

\bibitem{corbetComputingMulticoverBifiltration2023}
Ren{\'e} Corbet, Michael Kerber, Michael Lesnick, and Georg Osang.
\newblock Computing the {{Multicover Bifiltration}}.
\newblock {\em Discrete \& Computational Geometry}, 70(2):376--405, September
  2023.
\newblock \href {https://doi.org/10.1007/s00454-022-00476-8}
  {\path{doi:10.1007/s00454-022-00476-8}}.

\bibitem{dearingMinimumCoveringEuclidean2023}
P.~M. Dearing and Mark~E. Cawood.
\newblock The minimum covering {{Euclidean}} ball of a set of {{Euclidean}}
  balls in {$R^n$}.
\newblock {\em Annals of Operations Research}, 322(2):631--659, March 2023.
\newblock \href {https://doi.org/10.1007/s10479-022-05138-9}
  {\path{doi:10.1007/s10479-022-05138-9}}.

\bibitem{dyerSimpleHeuristicPcentre1985}
M.E Dyer and A.M Frieze.
\newblock A simple heuristic for the p-centre problem.
\newblock {\em Operations Research Letters}, 3(6):285--288, February 1985.
\newblock \href {https://doi.org/10.1016/0167-6377(85)90002-1}
  {\path{doi:10.1016/0167-6377(85)90002-1}}.

\bibitem{edelsbrunnerUnionBallsIts1995}
Herbert Edelsbrunner.
\newblock The union of balls and its dual shape.
\newblock {\em Discrete \& Computational Geometry}, 13:415--440, 1995.
\newblock \href {https://doi.org/10.1007/BF02574053}
  {\path{doi:10.1007/BF02574053}}.

\bibitem{edelsbrunnerComputationalTopologyIntroduction2010}
Herbert Edelsbrunner and John~L. Harer.
\newblock {\em Computational Topology: An Introduction}.
\newblock American Mathematical Society, Providence, R.I, 2010.
\newblock URL: \url{https://doi.org/10.1090/mbk/069}.

\bibitem{edelsbrunnerMulticoverPersistenceEuclidean2018conf}
Herbert Edelsbrunner and Georg Osang.
\newblock The {{Multi-cover Persistence}} of {{Euclidean Balls}}.
\newblock In {\em 34th {{International Symposium}} on {{Computational
  Geometry}}, {{SoCG}} 2018}, {{LIPIcs}}, page 14 pages, 2018.
\newblock \href {https://doi.org/10.4230/LIPICS.SOCG.2018.34}
  {\path{doi:10.4230/LIPICS.SOCG.2018.34}}.

\bibitem{edelsbrunnerMultiCoverPersistenceEuclidean2021}
Herbert Edelsbrunner and Georg Osang.
\newblock The {{Multi-Cover Persistence}} of {{Euclidean Balls}}.
\newblock {\em Discrete \& Computational Geometry}, 65(4):1296--1313, June
  2021.
\newblock \href {https://doi.org/10.1007/s00454-021-00281-9}
  {\path{doi:10.1007/s00454-021-00281-9}}.

\bibitem{fischerSmallestEnclosingBall2004}
Kaspar Fischer and Bernd G{\"a}rtner.
\newblock The smallest enclosing ball of balls: Combinatorial structure and
  algorithms.
\newblock {\em International Journal of Computational Geometry \&
  Applications}, 14(04n05):341--378, October 2004.
\newblock \href {https://doi.org/10.1142/S0218195904001500}
  {\path{doi:10.1142/S0218195904001500}}.

\bibitem{cgal:fghhs-bv-24a}
Kaspar Fischer, Bernd G{\"a}rtner, Thomas Herrmann, Michael Hoffmann, and Sven
  Sch{\"o}nherr.
\newblock Bounding volumes.
\newblock In {\em {{CGAL}} User and Reference Manual}. CGAL Editorial Board,
  5.6.1 edition, 2024.
\newblock URL:
  \url{https://doc.cgal.org/5.6.1/Manual/packages.html#PkgBoundingVolumes}.

\bibitem{gonzalezClusteringMinimizeMaximum1985}
Teofilo~F. Gonzalez.
\newblock Clustering to minimize the maximum intercluster distance.
\newblock {\em Theoretical Computer Science}, 38:293--306, 1985.
\newblock \href {https://doi.org/10.1016/0304-3975(85)90224-5}
  {\path{doi:10.1016/0304-3975(85)90224-5}}.

\bibitem{guibasWitnessedKDistance2013}
Leonidas Guibas, Dmitriy Morozov, and Quentin M{\'e}rigot.
\newblock Witnessed k-{{Distance}}.
\newblock {\em Discrete \& Computational Geometry}, 49(1):22--45, January 2013.
\newblock \href {https://doi.org/10.1007/s00454-012-9465-x}
  {\path{doi:10.1007/s00454-012-9465-x}}.

\bibitem{guibasWitnessedKdistance2011}
Leonidas~J. Guibas, Quentin M{\'e}rigot, and Dmitriy Morozov.
\newblock Witnessed k-distance.
\newblock In {\em Proceedings of the Twenty-Seventh Annual Symposium on
  {{Computational}} Geometry, {{SoCG}} 2011}, pages 57--64, June 2011.
\newblock \href {https://doi.org/10.1145/1998196.1998205}
  {\path{doi:10.1145/1998196.1998205}}.

\bibitem{har-peledFastConstructionNets2006a}
Sariel {Har-Peled} and Manor Mendel.
\newblock Fast {{Construction}} of {{Nets}} in {{Low Dimensional Metrics}}, and
  {{Their Applications}}.
\newblock {\em SIAM Journal on Computing}, 35(5):1148--1184, January 2006.
\newblock \href {https://arxiv.org/abs/cs/0409057} {\path{arXiv:cs/0409057}},
  \href {https://doi.org/10.1137/S0097539704446281}
  {\path{doi:10.1137/S0097539704446281}}.

\bibitem{hellmerDensitySensitiveBifiltered2024}
Niklas Hellmer and Jan Spali{\'n}ski.
\newblock Density {{Sensitive Bifiltered Dowker Complexes}} via {{Total
  Weight}}, September 2024.
\newblock URL: \url{http://arxiv.org/abs/2405.15592}, \href
  {https://arxiv.org/abs/2405.15592} {\path{arXiv:2405.15592}}.

\bibitem{lesnickNerveModelsSubdivision2024}
Michael Lesnick and Ken McCabe.
\newblock Nerve {{Models}} of {{Subdivision Bifiltrations}}, June 2024.
\newblock URL: \url{http://arxiv.org/abs/2406.07679}, \href
  {https://arxiv.org/abs/2406.07679} {\path{arXiv:2406.07679}}.

\bibitem{lesnickSparseApproximationSubdivisionRips2024}
Michael Lesnick and Kenneth McCabe.
\newblock Sparse {{Approximation}} of the {{Subdivision-Rips Bifiltration}} for
  {{Doubling Metrics}}, August 2024.
\newblock URL: \url{http://arxiv.org/abs/2408.16716}, \href
  {https://arxiv.org/abs/2408.16716} {\path{arXiv:2408.16716}}.

\bibitem{lesnickInteractiveVisualization2D2015}
Michael Lesnick and Matthew Wright.
\newblock Interactive {{Visualization}} of 2-{{D Persistence Modules}},
  December 2015.
\newblock URL: \url{http://arxiv.org/abs/1512.00180}, \href
  {https://arxiv.org/abs/1512.00180} {\path{arXiv:1512.00180}}.

\bibitem{matousekSubexponentialBoundLinear1996}
Ji{\v r}{\'i} Matou{\v s}ek, Micha Sharir, and Emo Welzl.
\newblock A {{Subexponential Bound}} for {{Linear Programming}}.
\newblock {\em Algorithmica}, 16:498--516, 1996.
\newblock \href {https://doi.org/10.1007/BF01940877}
  {\path{doi:10.1007/BF01940877}}.

\bibitem{niyogiFindingHomologySubmanifolds2008}
Partha Niyogi, Stephen Smale, and Shmuel Weinberger.
\newblock Finding the {{Homology}} of {{Submanifolds}} with {{High Confidence}}
  from {{Random Samples}}.
\newblock {\em Discrete \& Computational Geometry}, 39(1-3):419--441, March
  2008.
\newblock \href {https://doi.org/10.1007/s00454-008-9053-2}
  {\path{doi:10.1007/s00454-008-9053-2}}.

\bibitem{phillipsGeometricInferenceKernel2015}
Jeff~M. Phillips, Bei Wang, and Yan Zheng.
\newblock Geometric {{Inference}} on {{Kernel Density Estimates}}.
\newblock In {\em 31st {{International Symposium}} on {{Computational
  Geometry}}, {{SoCG}} 2015}, volume~34, pages 857--871. Schloss Dagstuhl --
  Leibniz-Zentrum f{\"u}r Informatik, 2015.
\newblock \href {https://doi.org/10.4230/LIPICS.SOCG.2015.857}
  {\path{doi:10.4230/LIPICS.SOCG.2015.857}}.

\bibitem{quillenHigherAlgebraicKtheory}
Daniel Quillen.
\newblock Higher algebraic {{K-theory}}: {{I}}.

\bibitem{quillenHomotopyPropertiesPoset1978}
Daniel Quillen.
\newblock Homotopy properties of the poset of nontrivial p-subgroups of a
  group.
\newblock {\em Advances in Mathematics}, 28(2):101--128, May 1978.
\newblock \href {https://doi.org/10.1016/0001-8708(78)90058-0}
  {\path{doi:10.1016/0001-8708(78)90058-0}}.

\bibitem{rolleStableConsistentDensityBased2024}
Alexander Rolle and Luis Scoccola.
\newblock Stable and {{Consistent Density-Based Clustering}} via
  {{Multiparameter Persistence}}.
\newblock {\em Journal of Machine Learning Research}, 25(258):1--74, 2024.
\newblock URL: \url{http://jmlr.org/papers/v25/21-1185.html}.

\bibitem{rosenkrantzAnalysisSeveralHeuristics1977}
Daniel~J. Rosenkrantz, Richard~E. Stearns, and Philip~M. Lewis.
\newblock An analysis of several heuristics for the traveling salesman problem.
\newblock {\em SIAM Journal on Computing}, 6(3):563--581, 1977.

\bibitem{scoccolaPersistablePersistentStable2023}
Luis Scoccola and Alexander Rolle.
\newblock Persistable: Persistent and stable clustering.
\newblock {\em Journal of Open Source Software}, 8(83):5022, March 2023.
\newblock \href {https://doi.org/10.21105/joss.05022}
  {\path{doi:10.21105/joss.05022}}.

\bibitem{sheehyGreedypermutations}
Donald~R. Sheehy.
\newblock Greedypermutations.
\newblock URL: \url{https://github.com/donsheehy/greedypermutation}.

\bibitem{sheehyLinearsizeApproximationsVietorisrips2012}
Donald~R. Sheehy.
\newblock Linear-size approximations to the {{Vietoris-Rips}} filtration.
\newblock In {\em Proceedings of the Twenty-Eighth Annual Symposium on
  {{Computational}} Geometry, {{SoCG}} 2012}, pages 239--248, June 2012.
\newblock \href {https://doi.org/10.1145/2261250.2261286}
  {\path{doi:10.1145/2261250.2261286}}.

\bibitem{sheehyMulticoverNerveGeometric2012}
Donald~R. Sheehy.
\newblock A {{Multicover Nerve}} for {{Geometric Inference}}.
\newblock In {\em Proceedings of the 24th {{Canadian Conference}} on
  {{Computational Geometry}}}, page~5, 2012.
\newblock URL: \url{https://donsheehy.net/research/sheehy12multicover.pdf}.

\bibitem{sheehyLinearSizeApproximationsVietoris2013}
Donald~R. Sheehy.
\newblock Linear-{{Size Approximations}} to the {{Vietoris}}--{{Rips
  Filtration}}.
\newblock {\em Discrete \& Computational Geometry}, 49(4):778--796, June 2013.
\newblock \href {https://doi.org/10.1007/s00454-013-9513-1}
  {\path{doi:10.1007/s00454-013-9513-1}}.

\bibitem{sheehySparseDelaunayFiltration2021}
Donald~R. Sheehy.
\newblock A {{Sparse Delaunay Filtration}}.
\newblock In {\em 37th {{International Symposium}} on {{Computational
  Geometry}}, {{SoCG}} 2021}, volume 189, pages 58:1--58:16. Schloss Dagstuhl
  -- Leibniz-Zentrum f{\"u}r Informatik, 2021.
\newblock \href {https://doi.org/10.4230/LIPICS.SOCG.2021.58}
  {\path{doi:10.4230/LIPICS.SOCG.2021.58}}.

\bibitem{smidWeakGapProperty2009}
Michiel Smid.
\newblock The {{Weak Gap Property}} in {{Metric Spaces}} of {{Bounded Doubling
  Dimension}}.
\newblock In Susanne Albers, Helmut Alt, and Stefan N{\"a}her, editors, {\em
  Efficient {{Algorithms}}}, volume 5760, pages 275--289. Springer Berlin
  Heidelberg, Berlin, Heidelberg, 2009.
\newblock \href {https://doi.org/10.1007/978-3-642-03456-5_19}
  {\path{doi:10.1007/978-3-642-03456-5_19}}.

\bibitem{sylvesterQuestionGeometrySituation1857}
J.~J. Sylvester.
\newblock A question on the geometry of situation.
\newblock {\em Quarterly Journal of Pure and Applied Mathematics}, 1:79, 1857.
\newblock URL:
  \url{http://resolver.sub.uni-goettingen.de/purl?PPN600494829_0001}.

\bibitem{welzlSmallestEnclosingDisks1991}
Emo Welzl.
\newblock Smallest enclosing disks (balls and ellipsoids).
\newblock In Hermann Maurer, editor, {\em New {{Results}} and {{New Trends}} in
  {{Computer Science}}}, volume 555, pages 359--370, Berlin/Heidelberg, 1991.
  Springer-Verlag.
\newblock \href {https://doi.org/10.1007/BFb0038202}
  {\path{doi:10.1007/BFb0038202}}.

\bibitem{zurcherSmallestEnclosingBall2007}
Samuel Z{\"u}rcher.
\newblock Smallest {{Enclosing Ball}} for a {{Point Set}} with {{Strictly
  Convex Level Sets}}.
\newblock Master's thesis, ETH Zurich, Institute of Theoretical Computer
  Science, 2007.
\newblock URL: \url{https://samzurcher.ch/assets/pdf/master-thesis.pdf}.

\end{thebibliography}

\end{document}